\documentclass[12pt]{article}

\usepackage{graphicx,epsfig}

 \usepackage{amsfonts,amssymb}
\usepackage{amsmath}
\usepackage{youngtab}

\usepackage{amsthm}
\theoremstyle{plain}% default
\newtheorem{thm}{Theorem}%[section]

\newtheorem{prop}[thm]{Proposition}

\theoremstyle{definition}

\newtheorem{rul}{Rule}%[section]
\theoremstyle{remark}

\usepackage{bm}
\usepackage{color}
\definecolor{dark-green}{rgb}{0,0.7,0}
\definecolor{dark-blue}{rgb}{0,0.2,0.5}
\definecolor{med-blue}{rgb}{0,0.7,1}
\definecolor{mblue}{rgb}{0,0.2,1}
\definecolor{cnc}{rgb}{0.8,0,0}
\definecolor{light-red}{rgb}{1,0.8,0.8}
\definecolor{dark-yellow}{rgb}{1,0.8,0}
\definecolor{light-blue}{rgb}{0.8,0.9,1}
\definecolor{verylight-blue}{rgb}{0.93,0.95,1}
\definecolor{light-yellow}{rgb}{1,0.9,0.8}
\definecolor{grey}{gray}{0.88}
\definecolor{new-green}{rgb}{0.5,0.5,0.6}

\def\a{\alpha} \def\b{\beta} \def\g{\gamma} \def\d{\delta}
 \def\ve{\varepsilon} 
  
 \def\S{{\cal{S}}} \def\A{{\cal{A}}} 
  \def\i{i_1} \def\j{i_2}
\def\k{i_3} \def\l{i_4}
\begin{document}

%\begin{footnotesize}
\title{The constitutive tensor of linear elasticity: its
  decompositions, Cauchy relations, null Lagrangians, and
  wave propagation}
\author{Yakov Itin\\
  Inst. Mathematics, Hebrew Univ. of Jerusalem, \\E.J.\ Safra Campus,
  Givat Ram,
  Jerusalem 91904, Israel,\\
  and Jerusalem College of Technology, 21 Havaad Haleumi, POB 16031,\\
  Jerusalem 91160, Israel,
  email: itin@math.huji.ac.il\\ \\
  Friedrich W.~Hehl\\
  Inst.\ Theor.\ Physics, Univ.\ of Cologne,
  50923 K\"oln, Germany,\\
  and Dept.\ Physics \& Astron., Univ.\ of Missouri,\\ Columbia, MO
  65211, USA, email: hehl@thp.uni-koeln.de}

\date{05 August 2012, {\it file DecompElasticity17.tex}}
\maketitle
%\begin{footnotesize}
\begin{abstract} In linear anisotropic elasticity, the elastic
  properties of a medium are described by the fourth rank elasticity
  tensor $C$. The decomposition of $C$ into a partially symmetric
  tensor $M$ and a partially antisymmetric tensors $N$ is often used
  in the literature. An alternative, less well-known decomposition,
  into the completely symmetric part $S$ of $C$ plus the reminder $A$,
  turns out to be {\it irreducible} under the 3-dimensional general
  linear group. We show that the $S\!A$-decomposition is unique,
  irreducible, and preserves the symmetries of the elasticity
  tensor. The $M\!N$-decomposition fails to have these desirable
  properties and is such inferior from a physical point of
  view. Various applications of the $S\!A$-decomposition are
  discussed: the Cauchy relations (vanishing of $A$), {the
    non-existence of elastic} null Lagrangians, the decomposition of
  the elastic energy and of the acoustic wave propagation. The
  acoustic or Christoffel tensor is split in a Cauchy and a non-Cauchy
  part. The Cauchy part governs the longitudinal wave propagation. We
  provide explicit examples of the effectiveness of the
  $S\!A$-decomposition. A complete class of anisotropic media is
  proposed that allows pure polarizations in arbitrary directions,
  similarly as in an isotropic medium.

\end{abstract}

{\it {Key index words:} {anisotropic} elasticity tensor,
  irreducible decomposition, Cauchy relations, null Lagrangians,
  acoustic tensor}

%\end{footnotesize}
\pagestyle{myheadings} \markboth{Yakov Itin} {Y.\ Itin and
  F.W.~Hehl\hspace{90pt} The constitutive tensor of linear
  elasticity}

\maketitle
%\tableofcontents
%\end{document}
%%%%%%%%%%%%%%%%%%%%%%%%
\section{Introduction  {and summary of the results}}
%%%%%

Consider an arbitrary point $P_o$, with coordinates $x^i_o$, in an
undeformed body. Then we deform the body and the same material point
is named $P$, with coordinates $x^i$. The position of $P$ is uniquely
determined by its initial position $P_o$, that is,
$x^i=x^i\left(x^j_o\right)$.  The {\it displacement} vector
$u=u^i \partial_i$ is then defined as
\begin{equation}\label{strain1}
u^i\left(x^j_o\right)=x^i\left(x^j_o\right)-x^i_o\qquad (i,j,\dots=1,2,3)\,.
\end{equation}
We distinguish between lower (covariant) and upper
  (contravariant) indices in order to have the freedom to change to 
arbitrary coordinates if necessary.

\subsection{Strain and stress}

The deformation and the stress state of an elastic body is, within
{\it linear} elasticity theory, described by means of the {strain}
tensor $\varepsilon_{ij}$ and the {stress} tensor $\sigma^{ij}$. The
strain tensor as well as the stress tensor are both symmetric, that
is, $\varepsilon_{[ij]}:=\frac 12 (\varepsilon_{ij}
-\varepsilon_{ji})=0$ and $\sigma^{[ij]}=0$, see Love
(1927)\cite{Love}, Landau \& Lifshitz (1986)\cite{landau}, Hauss\"uhl
(2007)\cite{Haus}, Marsden \& Hughes (1983)\cite{Marsden}, and
Podio-Guidugli (2000)\cite{Podio-Guidugli}; thus, $\varepsilon_{ij}$
and $\sigma^{ij}$ both have 6 independent components.

The strain tensor can be expressed in terms of the displacement vector
via 
\begin{equation}\label{strain}
\varepsilon_{ij}=g_{ki}\partial_j u^k+g_{kj}\partial_i u^k=:2
g_{k(i}\partial_{j)} u^k\,,
\end{equation}
whereupon $g_{ij}$ is the metric of the three-dimensional (3d)
Euclidean background space and $\partial_i:=\partial /\partial
x^i_o$. In Cartesian coordinates, we have
$\varepsilon_{ij}=2\partial_{(i}u_{j)}$, with $u_i:=g_{ik}u^k$.  

The stress tensor fulfills the momentum law
\begin{equation}\label{momentum}
 \partial_j\sigma^{ij}+\rho b^i=\rho \ddot{u}^i\,;
\end{equation}
here $\rho$ is the mass density, $b^i$ the body force density, and a
dot denotes the time derivative.

\subsection{Elasticity tensor}

The constitutive law in linear elasticity for a homogeneous
anisotropic body, the generalized Hooke law, postulates a linear
relation between the two second-rank tensor fields, the stress
$\sigma^{ij}$ and the strain $\varepsilon_{kl}$:
\begin{equation}\label{hooke}
\sigma^{ij}=C^{ijkl}\,\varepsilon_{kl}\,.
\end{equation}
The {elasticity tensor} $C^{ijkl}$ has the physical dimension
of a stress, namely {\it force/area.} Hence, in the International
System of Units (SI), the frame components of $C^{ijkl}$ are measured in
{\it pascal} P, with P:=N/m$^2$.

In 3d, a generic fourth-order tensor has 81 independent components. It
can be viewed as a generic $9\times 9$ matrix.  Since
$\varepsilon_{ij}$ and $\sigma_{ij}$ are symmetric, certain symmetry
relations hold also for the elasticity tensor. Thus,
\begin{equation}\label{s1-sym}
  \varepsilon_{[kl]}=0\,\qquad \Longrightarrow\qquad C^{ij[kl]}=0\,.
\end{equation}
This is the so-called {\it right minor symmetry}.  Similarly and
independently, the symmetry of the stress tensor yields the so-called
{\it left minor symmetry},
\begin{equation}\label{s2-sym}
\sigma^{[ij]}=0\,\qquad \Longrightarrow\qquad C^{[ij]kl}=0\,.
\end{equation}
Both minor symmetries, (\ref{s1-sym}) and (\ref{s2-sym}), are assumed
to hold simultaneously. Accordingly, the tensor $ C^{ijkl}$ can be
represented by a $6\times 6$ matrix with 36 independent components.

The energy density of a deformed material is expressed as
$W={\scriptstyle\frac 12 }\,\sigma^{ij}\varepsilon_{ij}$. When the
Hooke law is substituted, this expression takes the form
\begin{equation}\label{energy1a}
W=\frac 12 \,C^{ijkl}\varepsilon_{ij}\varepsilon_{kl}\,.
\end{equation}
The right-hand side of (\ref{energy1a}) involves only those
combinations of the elasticity tensor components which are symmetric
under permutation of the first and the last pairs of indices.  In
order to prevent the corresponding redundancy in the components of
$C^{ijkl}$, the so-called {\it major symmetry},
\begin{equation}\label{paircom}
  C^{ijkl}- C^{klij}=0\,
\end{equation}
is assumed. Therefore, the $6\times 6$ matrix becomes symmetric and
only 21 independent components of $C^{ijkl}$ are left
over.

As a side remark we mention that the components of a
  tensor are always measured with respect to a local frame
$e_\a=e^i{}_\a \partial_i$; here $\a=1,2,3$ numbers the three linearly
independent legs of this frame, the triad. Dual to this frame is the
coframe $\vartheta^\b=e_j{}^\b dx^j$, see Schouten
(1989)\cite{SchoutenPhysicist} and Post (1997)\cite{Post}. For the
elasticity tensor, the components with respect to such a local coframe
are $C^{\a\b\g\d}:=e_i{}^\a e_j{}^\b e_k{}^\g e_l{}^\d C^{ijkl}$. They
are called the physical components of $C$. For simplicity,
  we will not set up a frame formalism since it doesn't provide
  additional insight in the decomposition of the elasticity tensor,
  that is, we will use coordinate frames $\partial_i$ in the rest of
  our article.

Because of (\ref{s1-sym}), (\ref{s2-sym}), and
  (\ref{paircom}), we have the following \medskip

  \noindent{\bf Definition:} A {\it 4th rank tensor} of type
  $\binom{4}{0}$ qualifies to describe anisotropic elasticity if

 (i) its physical components carry the dimension of
force/area (in SI pascal),

 (ii) it obeys the left and right minor symmetries,

 (iii) and it obeys the major symmetry.

 \noindent It is then called elasticity tensor (or elasticity or
 stiffness) and, in general, denoted by $C^{ijkl}$. \medskip

 In Secs.~2.1 and 2.2 we translate our notation into that of Voigt,
 see Voigt (1928)\cite{Voigt} and Love (1927)\cite{Love}, and discuss
 the corresponding 21-dimensional vector space of all elasticity
 tensors, {see also Del Piero (1979)\cite{Piero}.

   Incidentally, in linear electrodynamics, see Post
   (1962)\cite{Post}, Hehl \& Obukhov (2003)\cite{Birkbook}, and Itin
   (2009)\cite{Itin:2009aa}, we have a 4-dimensional constitutive
   tensor $\chi^{\mu\nu\sigma\kappa}= -\chi^{\nu\mu\sigma\kappa}=
   -\chi^{\mu\nu\kappa\sigma} =\chi^{\sigma\kappa\mu\nu}$, with
   $\mu,\nu,...=0,1,2,3$. Surprisingly, this tensor corresponds also
   to a $6\times 6$ matrix, like $C^{ijkl}$ in elasticity. The major
   symmetry is the same, the minor symmetries are those of an
   antisymmetric pair of indices. }

\subsection{Decompositions of the 21 components elasticity tensor}

In Sec.~2.3, we turn first to an algebraic decomposition of $C^{ijkl}$
that is frequently discussed in the literature: the elasticity tensor
is decomposed into the sum of the two tensors $ M^{ijkl} :=
C^{i(jk)l}$ and $N^{ijkl} := C^{i[jk]l}$, which are symmetric or
antisymmetric in the middle pair of indices, respectively. We show
that $M$ and $N$ fulfill the major symmetry but {\it not} the minor
symmetries and that they can be further decomposed.  Accordingly, this
reducible decomposition does not correspond to a direct sum
decomposition of the vector space defined by $C$, as we will show in
detail. {Furthermore we show that the vector space of $M$
  is 21-dimensional and that of $N$ 6-dimensional.}

Often, in calculation within linear elasticity, the tensors $M$ and
$N$ emerge.  They are auxiliary quantities, but due to the lack of the
minor symmetries, they are {\it not} elasticities. Consequently, they
cannot be used to characterize a certain material elastically.  These
quantities are placeholders that are not suitable for a direct
physical interpretation.

Subsequently, in Sec.~2.4, we study the behavior of the physical
components of $C$ under the action of the general linear 3d real group
$GL(3,\mathbb R)$. The $GL(3,\mathbb R)$ commutes with the
permutations of tensor indices. This fact yields the well-known
relation between the action of $GL(3,\mathbb R)$ and the action of the
symmetry group $S_p$. Without restricting the generality of our
considerations, we will choose local coordinate frames $\partial_i$
for our considerations.
 
\vspace{20pt}
\begin{center}
\includegraphics[width=7truecm]{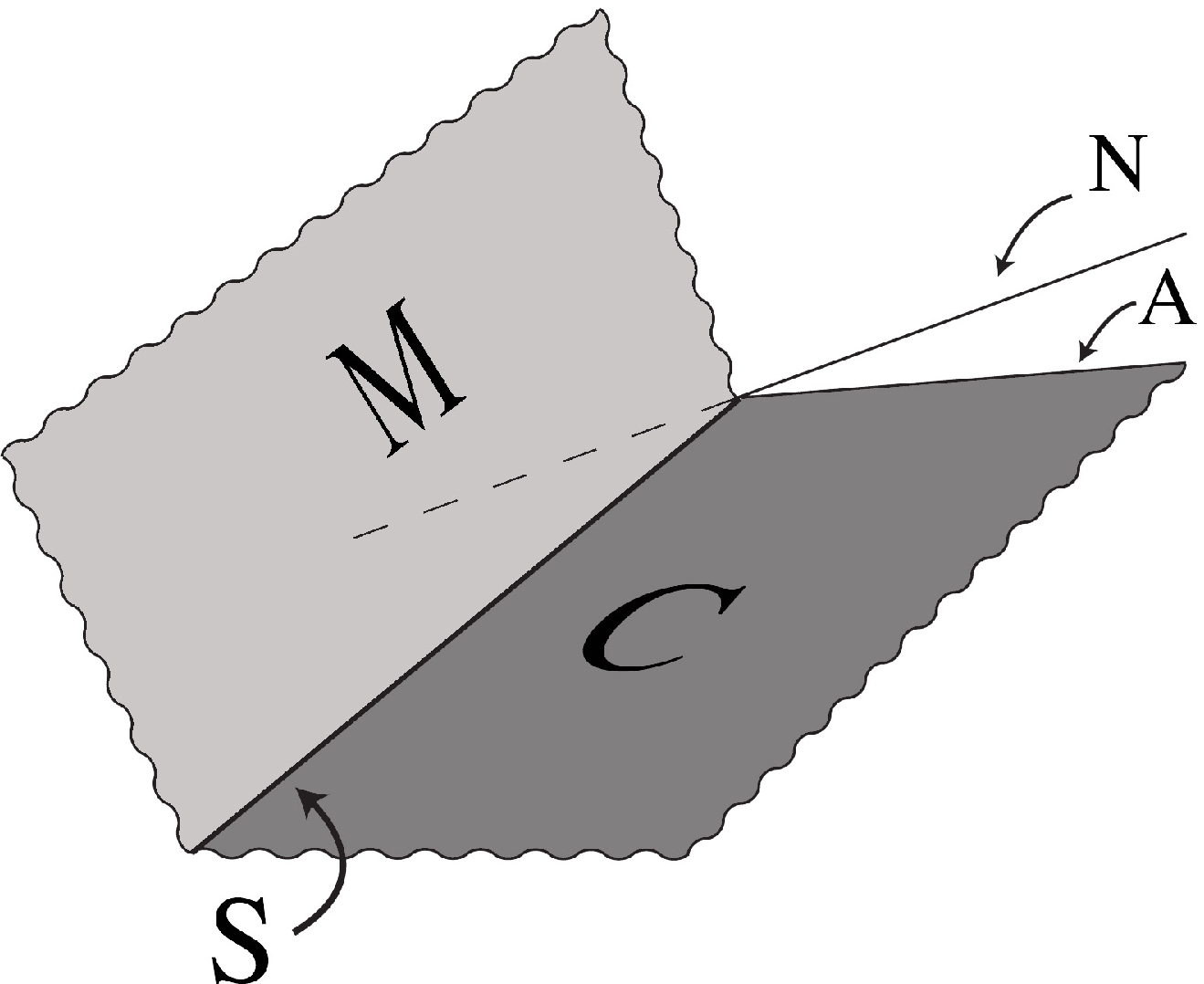}
\end{center}
  \noindent {\bf Figure 1.}
  {\it A tensor $T^{ijkl}$ of rank 4 in 3-dimensional (3d) space has
  $3^4=81$ independent components. The 3 dimensions of our image
  represent this 81d space. The plane $\mathbb C$ depicts the 21
  dimensional subspace of all possible elasticity (or stiffness)
  tensors. This space is span by its irreducible pieces, the 15d space
  of the totally symmetric elasticity $\mathbf S$ (a straight line)
  and the 6d space of the difference $A={\mathbb C}-{\mathbf S}$ (also
  depicted, I am sad to say, as a straight line).--- Oblique to
  $\mathbb C$ is the 21d space $\mathbb M$ of the reducible $M$-tensor
  and the 6d space of the reducible $N$-tensor. The $\mathbb C$
  ``plane'' is the only place where elasticities (stiffnesses) are at
  home. The spaces $\mathbb M$ and $N$ represent only elasticities,
  provided the Cauchy relations are fulfilled. Then, $A=N=0$ and
  $\mathbb M$ and $\mathbb C$ cut in the 15d space of $\mathbf
  S$. Notice that the spaces M and C are intersecting exactly on S.}
\bigskip

In this way, we arrive at an alternative decomposition of $C$ into two
pieces $S$ and $A$, which is {\it irreducible} under the action of the
$GL(3,\mathbb R)$.  The known device to study the action of $S_p$ is
provided by Young's tableaux technique. For the sake of completeness,
we will be briefly describe Young's technique in the Appendix. In an
earlier paper, see Hehl \& Itin (2002)\cite{Cauchy}, we discussed this
problem already, but here we will present rigorous proofs of all
aspects of this irreducible decomposition. {It turns our
  that the space of the $S$-tensor is 15-dimensional and that of the
  $A$-tensor 6-dimensional.}

In Sec.~2.5, we compare the reducible $M\!N$-decomposition {of
  Sec.~2.3} with the irreducible $S\!A$-decom\-posi\-tion {of
  Sec.~2.4} and will show that the latter one is definitely to be
preferred from a physical point of view. The formulas for the
transition between the $M\!N$- and the $S\!A$-decom\-posi\-tion are
collected in the Propositions 9 and 11. The irreducible 4th rank
tensor $A^{ijkl}$ can alternatively be represented by a symmetric 2nd
rank tensor $\Delta_{ij}$ (Proposition 10). We visualized our main
results with respect to the reducible and the irreducible
decompositions of the elasticity tensor in Figure 1, for details,
please see Sec.2.5.

{ The group $GL(3,\mathbb R)$, which we are using here,
  provides the basic, somewhat coarse-grained decomposition of the
  elasticity tensor. For finer types of irreducible decompositions
  under the {\it orthogonal} subgroup $SO(3)$ of $GL(3,\mathbb R)$,
  see, for example, Rychlewski (1984)\cite{Rychlewski}, Walpole
  (1984)\cite{Walpole}, Surrel (1993)\cite{Surrel}, Xiao
  (1998)\cite{Xiao}, and the related work of Backus
  (1970)\cite{Backus} and Baerheim (1993)\cite{Baerheim}. As a result
  of our presentation, the great number of invariants, which emerge in
  the latter case, can be organized into two subsets: one related to
  the $S$ piece and the other one to the $A$ piece if $C$.  }

\subsection{Physical applications and examples}

Having now the irreducible decomposition of $C$ at our command, Sec.~3
will be devoted to the physical applications. In Sec.3.1, we discuss
the Cauchy relations {and show that they correspond to the
  vanishing of one irreducible piece, namely to $A=0$ or,
  equivalently, to $\Delta=0$. As a consequence, the totally symmetric
  piece $S$ can be called the Cauchy part of the elasticity tensor
  $C$, whereas the $A$ piece measures the deviation from this Cauchy
  part; it is the non-Cauchy part of $C$. The reducible pieces $M$ and
  $N$ defy such an interpretation and are not useful for applications
  in physics.  In this sense, we can speak of two kinds of elasticity,
  a Cauchy type and a non-Cauchy type.

  In Sec.~3.2, this picture is brought to the elastic energy and the
  latter} decomposed in a Cauchy part and its excess, the ``non-Cauchy
part''. In other words, this distinction between two kinds of
elasticity is also reflected in the properties of the elastic energy.

A null Lagrangians is such an Lagrangian whose Euler-Lagrange
expression vanishes, see Crampin \& Saunders (2005)\cite{Crampin}; we
also speak of a ``pure divergence''.  In Sec.~3.3, elastic null
Lagrangians are addressed, and we critically evaluate the
literature. We show---in contrast to a seemingly widely held
view---that for an arbitrary anisotropic medium there doesn't exist an
elastic null Lagrangian. The expressions offered in the corresponding
literature are worthless as null Lagrangians, since they still depend
on {\it all} components of the elasticity tensor. We collect these
results in Proposition 12.

In Sec.~3.4, we define for acoustic wave propagation the Cauchy and
non-Cauchy parts of the Christoffel (or acoustic) tensor
$\Gamma^{ij}=C^{klj}n_kn_l/\rho$ ($n_k=$ unit wave covector, $\rho=$
mass density). We find some interesting and novel results for the
Christoffel tensor{(, see Propositions 13 and 14)}.  In Sec.~3.5, we
investigate the polarizations of the elastic wave. We show that the
longitudinal wave propagation is completely determined by the Cauchy
part of the Christoffel tensor, see Proposition 15. In Proposition 16
a new result is presented on the propagation of purely polarized
waves; {we were led to these investigations by following up
  some ideas about the interrelationship of the symmetry of the
  elasticity tensor and the Christoffel tensor in the papers of
  Alshits and Lothe (2004)\cite{Alshits} and B\'ona et al. (2004, 2007,
  2010)\cite{Bona2,Bona1,Bona}.}

In Sec.~4, we investigate examples, namely isotropic media (Sec.~4.1)
and media with cubic symmetry (Sec.~4.2). Modern technology allows
modeling composite materials with their effective elastic properties,
{see Tadmore \& Miller (2011)\cite{Tadmore}.} Irreducible
  decomposition of the elasticity tensor can be used as a guiding
  framework for prediction of certain features of these materials. As
  an example, we {presented in Proposition 17} a complete new class of
  anisotropic materials that allow pure polarizations {to propagate}
  in arbitrary directions, similarly as in isotropic materials.

\subsection{Notation}

We use here tensor analysis in 3d Euclidean space with explicit index
notation, see Sokolnikoff (1951)\cite{Sokol} and Schouten
(1954, 1989)\cite{SchoutenRicci,SchoutenPhysicist}. Coordinate
  (holonomic) indices are denoted by Latin letters $i,j,k,\dots$; they
  run over $1,2,3$. Since we allow arbitrary curvilinear coordinates,
  covariant and contravariant indices are used, that is, those in
  lower and in upper position, respectively, see Schouten
  (1989)\cite{SchoutenPhysicist}. Summation over repeated indices is
  understood. We abbreviate symmetrization and antisymmetrization over
  $p$ indices as follows:
\begin{eqnarray*}\label{symm}
  (i_1\,i_2\,\dots i_p)&:=&\frac{1}{p!}\{\sum
  \text{ over all permutations of }i_1\,i_2\,\dots i_p \}\,,\\
  \left[i_1\,i_2\,\dots i_p\right]&:=&\frac{1}{p!}\{\sum
  \text{over all even perms.} - \sum \text{over all odd perms.}\}\,.
\end{eqnarray*}
The Levi-Civita symbol is given by $\epsilon_{ijk}=+1,-1,0$, for even,
odd, and no permutation of the indices $123$, respectively; the
analogous is valid for $\epsilon^{ijk}$. The metric
$ds^2=g_{ij}dx^i\otimes dx^j$ has Euclidean signature. We can raise
and lower indices with the help of the metric. In linear elasticity
theory, tensor analysis is used, for example, in Love
(1927)\cite{Love}, Sokolnikoff (1956)\cite{SokolElast}, Landau \&
Lifshitz (1986)\cite{landau}, and Hauss\"uhl (2007)\cite{Haus}. For a
modern presentation of tensors as linear maps between corresponding
vector spaces, see Marsden \& Hughes (1983)\cite{Marsden},
Podio-Guidugli (2000)\cite{Podio-Guidugli}, {and Hetnarski
  \& Ignaczak (2011)\cite{Hetnarski}.}

In Hehl \& Itin (2002)\cite{Cauchy}, we denoted the elasticity moduli
differently. The quantities of our old paper translate as follows
into the present one: $^{(1)}{c}^{ijkl}\equiv S^{ijkl}$,
$^{(2)}{c}^{ijkl}\equiv A^{ijkl}$, $^{(2)}\hat{c}^{ijkl}\equiv
N^{ijkl}$, and ${c}^{i(jk)l}\equiv M^{ijkl}$.

\section{Algebra of the decompositions of the elasticity tensor}

\subsection{Elasticity tensor in Voigt's notation}

The standard ``shorthand'' notation of $C^{ijkl}$ is due to Voigt, see
Voigt (1928)\cite{Voigt} and Love (1927)\cite{Love}.  One identifies a
symmetric pair $\{ij\}$ of 3d indices with a multi-index $I$ that has
the range from $1$ to $6$:
\begin{equation}\label{voigt1} 11\to 1\,, \,22\to
2\,,\, 33\to 3\,,\, 23\to 4\,,\, 31\to 5\,,\, 12\to 6\,.
\end{equation}
Then the elasticity tensor is expressed as a symmetric $6\times 6$
matrix $C^{IJ}$.  Voigt's notation is only applicable since the minor
symmetries (\ref{s1-sym}) and (\ref{s2-sym}) are valid. Due to the
major symmetry, this matrix is symmetric, $C^{[IJ]}=0$. Explicitly,
we have
\begin{equation}\label{voigt2}
% \scriptsize{
\begin{bmatrix}
C^{1111} & C^{1122} & C^{1133} & C^{1123} & C^{1131} & C^{1112} \\
* & C^{2222} & C^{2233} & C^{2223} & C^{2231} & C^{2212} \\
* & * & C^{3333} & C^{3323} & C^{3331} & C^{3312} \\
* & * & * & C^{2323} & C^{2331} & C^{2312} \\
* & * & * & * & C^{3131} & C^{3112} \\
* & * & * & * & * & C^{1212}
     \end{bmatrix} \equiv \begin{bmatrix}
  C^{11} & C^{12} & C^{13} & C^{14} & C^{15} & C^{16} \\
* & C^{22} & C^{23} & C^{24} & C^{25} & C^{26} \\
* & * & C^{33} & C^{34} & C^{35} & C^{36} \\
* & * & * & C^{44} & C^{45} & C^{46} \\
* & * & * & * & C^{55} & C^{56} \\
* & * & * & * & * & C^{66}\, \end{bmatrix}.%}
\end{equation}
For general anisotropic materials, all the displayed components are
nonzero and independent of one another. The stars in both matrices
denote those entries that are dependent due to the symmetries of the
matrices.

\subsection{Vector space of the elasticity tensor}

The set of all generic elasticity tensors, that is, of all fourth rank
tensors with the minor and major symmetries, builds up a vector
space. Indeed, a linear combination of two such tensors, taken with
arbitrary real coefficients, is again a tensor with the same
symmetries. We denote this vector space by ${\mathcal C}$.
\begin{prop} For the space of elasticity tensors,
\begin{equation}
  {\rm dim}\,{\mathcal C}=21\,.
\end{equation}
\end{prop}
\begin{proof} The basis of ${\mathcal C}$ can be enumerated by the
  elements of the matrix $C^{IJ}$. For instance, the element
  $C^{11}=C^{1111}$ is related to the basis vector
\begin{equation}\label{bas1}
E_1=\partial_1\otimes\partial_1\otimes\partial_1\otimes\partial_1\,;
\end{equation}
the element $C^{12}=C^{1122}$ corresponds to the basis vector
\begin{equation}\label{bas2}
E_2=\frac12\left(\partial_1\otimes\partial_1\otimes\partial_2\otimes\partial_2+
\partial_2\otimes\partial_2\otimes\partial_1\otimes\partial_1\right)\,;
\end{equation}
the element $C^{13}=C^{1133}$ corresponds to the basis vector
\begin{equation}\label{bas2*}
  E_3=\frac12\left(\partial_1\otimes\partial_1\otimes\partial_3\otimes\partial_3+
\partial_3\otimes\partial_3\otimes\partial_1\otimes\partial_1\right)\,.
\end{equation}
To the element $C^{14}=C^{1123}$, we relate the basis vector
\begin{equation}\label{bas4}
E_4=\frac14\left(\partial_1\otimes\partial_1\otimes\partial_2\otimes\partial_3+
\partial_1\otimes\partial_1\otimes\partial_3\otimes\partial_2+
\partial_2\otimes\partial_3\otimes\partial_1\otimes\partial_1+
\partial_3\otimes\partial_2\otimes\partial_1\otimes\partial_1\right)\,.
\end{equation}
In this way, a set of vectors $\{E_1,\cdots,E_{21}\}$ is
constructed. Since there are no relations between the 21 components
$C^{IJ}$, all these vectors are linearly independent. Moreover every
elasticity tensor can be expanded as a linear combination of
$\{E_1,\cdots,E_{21}\}$. Thus a basis of of the vector space
${\mathcal C}$ consists of 21 vectors.
\end{proof}

\subsection{Reducible decomposition of $C^{ijkl}$}

\subsubsection{Definitions of  $M$ and $N$ and their symmetries}

In the literature on elasticity, a special decomposition of $C^{ijkl}$
into two tensorial parts is frequently used, see, for example, Cowin
(1989)\cite{Cowin1989}, Campanella \& Tonton (1994)\cite{Campanella},
Podio-Guidugli (2000)\cite{Podio-Guidugli}, Weiner
(2002)\cite{Weiner}, and Hauss\"uhl (2007)\cite{Haus}. It is obtained
by symmetrization and antisymmetrization of the elasticity tensor with
respect to its {\it two middle indices:}
\begin{eqnarray}\label{MNdef}
M^{ijkl}:=C^{i(jk)l}\,,\quad
  N^{ijkl}:=C^{i[jk]l}\,,\qquad\text{with}\qquad
C^{ijkl}=M^{ijkl}+ N^{ijkl}\,.
\end{eqnarray}
Sometimes the same operations are applied for the second and the
fourth indices. Due to the symmetries of the elasticity tensor, these
two procedures are equivalent to one another.

We recall that the elasticity tensor fulfills the left and right
minor symmetries and the major symmetry:
\begin{eqnarray}\label{minor}
  C^{[ij]kl}=0\,,\qquad C^{ij[kl]}=0\,;\qquad C^{ijkl}-C^{klij}=0\,.
\end{eqnarray}
\begin{prop}
The major symmetry holds for both tensors  $M^{ijkl}$ and $N^{ijkl}$.
\end{prop}
\begin{proof}
  We formulate the left-hand side of the major symmetries and
  substitute the definitions given in (\ref{MNdef}):
\begin{eqnarray}\label{MNsymmaj}
\hspace{-10pt}M^{ijkl}-M^{klij}\!&\!=\!&\!C^{i(jk)l}-C^{k(li)j}=\frac 12\left(C^{ijkl}+
  C^{ikjl}-C^{klij}-C^{kilj}\right)=0\,,\\ \label{MNsymmaj*}
N^{ijkl}-N^{klij}\!&\!=\!&\!C^{i[jk]l}-C^{k[li]j}=\frac 12\left(C^{ijkl}-
  C^{ikjl}-C^{klij}+C^{kilj}\right)=0\,.
\end{eqnarray}
\end{proof}
\begin{prop}
  In general, the minor symmetries do not hold for the tensors
  $M^{ijkl}$ and $N^{ijkl}$.
\end{prop}
\begin{proof}
  We formulate the left minor symmetries for $M$ and $N$ and use again
  the definitions from (\ref{MNdef}):
\begin{eqnarray}\label{MNsym}
  M^{[ij]kl}&=&{\frac 12}\left(C^{[ij]kl}+C^{[i|k|j]l}\right)=\frac 12
  C^{k[ij]l}=\frac 12 N^{kijl}\,,\\ 
   N^{[ij]kl}&=&{\frac 12}\left(C^{[ij]kl}-C^{[i|k|j]l} \right)
=-{\frac 12}C^{k[ij]l}=-{\frac 12}N^{kijl}\,;
\end{eqnarray}
here indices that are excluded from the (anti)symmetrization are
enclosed by vertical bars. Both expressions don't vanish in
general. Moreover, we are immediately led to $ M^{[ij]kl}=-
N^{[ij]kl}\ne 0$.

Using the major symmetry of Proposition 2, we recognize that the right
minor symmetries $M^{ij[kl]} =0$ and $N^{ij[kl]} = 0$ do not hold either,
since $ M^{ij[kl]}= M^{[kl]ij}$ and $ N^{ij[kl]}= N^{[kl]ij}$.
\end{proof}
Consequently, the tensors $M^{ijkl}$ and $N^{ijkl}$ do not belong to
the vector space ${\mathcal C}$ and cannot be written in Voigt's
notation.  Thus, these partial tensors $M$ and $N$ themselves cannot
serve as elasticity tensors for any material.

\subsubsection{Vector spaces of the  $N$- and $M$-tensors}

We will denote the set of all $N$-tensors by ${\mathcal N}$. It is a
vector space.  Indeed, $N^{ijkl}$ is defined as a fourth rank tensor
which is skew-symmetric in the middle indices and constructed from the
elasticity tensor. A linear combination of such tensors $\a N^{ijkl}
+\b\tilde {N}^{ijkl}$ will be also skew-symmetric. Moreover, it can be
constructed from the tensor $\a C^{ijkl} +\b\tilde {C}^{ijkl}$, which
satisfies the basic symmetries of the elasticity tensor.

A simplest way to describe a finite dimensional vector space, is to
write-down its basis. In the case of a tensor vector space, it is
enough to enumerate all the independent components of the tensor.
\begin{prop} For the space of $N$-tensors, ${\rm dim} \, {\cal N}=6$.
\end{prop}
\begin{proof}
  We can write-down explicitly six components of the ``antisymmetric''
  tensor $N^{ijkl}$ as
\begin{eqnarray}\label{asym-alter}
  &&N^{1122}=\frac 12 \left(C^{12}-C^{66}\right),\quad N^{1133}=
  \frac 12 \left(C^{13}-C^{55}\right),\quad N^{1123}=\frac 12 
  \left(C^{14}-C^{56}\right),\nonumber\\
  &&N^{2233}= \frac 12 \left(C^{23}-C^{44}\right),\quad N^{2231}=
  \frac 12 \left(C^{25}-C^{46}\right),\quad N^{1233}=
  \frac 12 \left(C^{36}-C^{45}\right).
\end{eqnarray}
All other components vanish or differ from the components given in
(\ref{asym-alter}) only in sign.  Since all the components $C^{IJ}$
with $I\le J$ are assumed to be independent, then the components of
$N^{ijkl}$ of (\ref{asym-alter}) are also independent.
\end{proof}

The set of all $M$-tensors is also a vector space, which we denote by
${\cal M}$.
\begin{prop} For the space of $M$-tensors, ${\rm dim} \, {\cal M}=21$.
\end{prop}
\begin{proof}
  The dimension of the vector space ${\cal M}$ can also be calculated
  by considering the independent components of a generic tensor
  $M^{ijkl}$. We find,
\begin{eqnarray}\label{sym-alter}
  &&M^{1111}=C^{11},\quad M^{1113}=C^{15},\quad M^{1112}=C^{16},\quad 
  M^{2222}=C^{22}\,,\nonumber\\
  &&M^{2221}=C^{26},\quad M^{3333}=C^{33},\quad M^{3332}=C^{34},\quad 
  M^{3331}=C^{35},\nonumber\\
  &&M^{2331}=C^{45},\quad M^{3221}=C^{46},\quad M^{3113}=C^{55},\quad 
  M^{3112}=C^{56}\,,\nonumber\\
  &&M^{1122}=\frac 12 \left(C^{12}+C^{66}\right),\quad M^{1133}=
  \frac 12 \left(C^{13}+C^{55}\right),\quad M^{2223}=C^{24},\nonumber\\
  &&M^{1123}=\frac 12 \left(C^{14}+C^{56}\right),\quad M^{2233}=
  \frac 12 \left(C^{23}+C^{44}\right),\quad M^{2332}=C^{44}\,,\nonumber\\
  &&M^{2231}=\frac 12 \left(C^{25}+C^{46}\right),\quad M^{1233}=
  \frac 12 \left(C^{36}+C^{45}\right),\quad M^{1221}=C^{66}\,.
\end{eqnarray}
All these 21 components are linearly independent, thus the dimension
of the vector space ${\cal M}$ is at least 21.  However, since every
element of ${\cal M}$ is defined in terms of 21 independent elastic
constants $C^{IJ}$, the dimension of ${\cal M}$ cannot be greater than
21.
\end{proof}

\subsubsection{Algebraic properties of   $M$ and $N$ tensors}

Observe some principal features of the tensors $M^{ijkl}$ and
$N^{ijkl}$:\medskip

%\noindent
{\it (i)~Inconsistency.} In general, a certain component of
$C^{ijkl}$, say $C^{1223}$, can be expressed in different ways in terms of
the components of $M^{ijkl}$ a $N^{ijkl}$:
\begin{equation}\label{xx1}
  C^{1223}\,\stackrel{(22)}{=}\,M^{1223}+\underbrace{N^{1223}}_{=0}
  \,\stackrel{\text{maj}}{=}\,M^{2312}\stackrel{\text{sym}}{=}M^{2132}
  \,\stackrel{(23)}{=}\,C^{46}\,.
\end{equation}
On the other hand, we have
\begin{equation}\label{alt5}
C^{1223}=C^{2123}\,\stackrel{(22)}{=}\,M^{2123}+N^{2123}\,.
\end{equation}
With
\begin{equation}\label{alt5x}
M^{2123}=\frac 12\left(C^{46}+C^{25}\right)\qquad\text{and}\qquad 
 N^{2123}=\frac 12\left(C^{46}-C^{25}\right)\ne 0\,,
\end{equation}
we recover the result in (\ref{xx1}), but is was achieved with the help
of a non-vanishing component of $N$.

%\noindent
{\it (ii)~Reducibility.} Since, in general, the tensor $M^{ijkl}$ is
not completely symmetric, a finer decomposition is possible,
\begin{equation}\label{alt5y}
  M^{ijkl}=M^{(ijkl)}+K^{ijkl}\,. 
\end{equation}
Accordingly, $C^{ijkl}$ can be decomposed into three tensorial pieces:
\begin{equation}\label{alt5z}
 C^{ijkl}=M^{(ijkl)}+K^{ijkl}+N^{ijkl}\,.
\end{equation}
%\noindent

{\it (iii)~Vector spaces.} The ``partial'' vector spaces, ${\cal
  M}$ and ${\cal N}$, are not subspaces of the vector space ${\cal
  C}$ and their sum ${\cal M}+{\cal N}$ is not equal to ${\cal C}$.
\medskip

Thus, the $M\!N$-decomposition is problematic from an algebraic point
of view. Our aim is to present an alternative irreducible
decomposition with better algebraic properties.
\bigskip

\subsection{Irreducible decomposition of  $C^{ijkl}$}

\subsubsection{Definitions of  $S$ and $A$ and their symmetries}

In Eq.(\ref{tab1}) of the Appendix, we decomposed a fourth rank
tensor irreducibly. Let us apply it to the elasticity tensor
$C^{ijkl}$. Since the dimension of 3d space is less than the rank of
the tensor, the last diagram in (\ref{tab1}), representing
$C^{[ijkl]}$, is identically zero. Also the minor symmetries remove
some of the diagrams. Dropping the diagrams which are antisymmetric in
the pairs of the indices $(i_1,i_2)$ and $(i_3,i_4)$, we are
eventually left with the decomposition
\begin{eqnarray}\label{tab2}
  \Yvcentermath1\young(\i)\otimes\young(\j)\otimes\young(\k)\otimes
  \young(\l)&=&\Yvcentermath1 \frac{1}{4!}\,\young({\i}{\j}{\k}{\l})
  \Yvcentermath1+\b\Bigg(\young({\i}{\j}{\k},{\l})+ 
  \young({\i}{\j}{\l},{\k})\Bigg)
  \Yvcentermath1+\g \,\young({\i}{\j},{\k}{\l})\,.
\end{eqnarray}
Let us now apply the major symmetry. It can be viewed as pair of
simultaneous permutations $\i\!\leftrightarrow\! \k$ and
$\j\!\leftrightarrow\! \l$. The first and the last diagrams in
(\ref{tab2}) are invariant under those transformations. The two
diagrams in the middle change their signs and are thus identically
zero. Thus, we are left with irreducible parts:
\begin{eqnarray}\label{tab2x}
\Yvcentermath1\young(\i)\otimes\young(\j)\otimes\young(\k)
\otimes\young(\l)&=&\Yvcentermath1 \frac{1}{4!}\,\young({\i}{\j}{\k}{\l})
\Yvcentermath1+\g \,\young({\i}{\j},{\k}{\l})\,.
\end{eqnarray}
In correspondence with the first table, the first subtensor of
$C^{ijkl}$ is derived by complete symmetrization of its indices
\begin{eqnarray}\label{sub1}
S^{ijkl}=\Big(I+(i_1,i_2)+(i_1,i_3)+\cdots+(i_1,i_2,i_3)+
\cdots+(i_1,i_2,i_3,i_4)\Big)C^{ijkl}\,.
\end{eqnarray}
Here the parentheses  denote the cycles of permutations. 
Consequently,
\begin{eqnarray}\label{sub1x}
%\boxed
S^{ijkl}:=C^{(ijkl)}=\frac{1}{4!}\Big(C^{ijkl}+C^{jikl}+C^{kjil}+\cdots+C^{kijl}+\cdots+
C^{lijk}\Big)\,.
\end{eqnarray}
On the right hand side, we have a sum of $4!=24$ terms with all
possible orders of the indices. If we take into account the symmetries
(\ref{s1-sym}), (\ref{s2-sym}), and (\ref{paircom}) of $C^{ijkl}$, we
can collect the terms:
\begin{equation}\label{firsty}
S^{ijkl}=C^{(ijkl)}=\frac 13(C^{ijkl}+C^{iklj}+ C^{iljk})\,.
\end{equation}

According to (\ref{tab2x}), the second irreducible piece of $C^{ijkl}$
can be defined as
\begin{equation}\label{A=C-S}
% \boxed
A^{ijkl}:=C^{ijkl}-C^{(ijkl)}=C^{ijkl}-S^{ijkl}\,.
\end{equation}
This result can also be derived by evaluating the last diagram in
(\ref{tab2x}).
Substitution of (\ref{firsty}) into the right-hand side of
(\ref{A=C-S}) yields
\begin{equation}\label{firsty+}
A^{ijkl}=\frac 13(2C^{ijkl}-C^{ilkj}-C^{iklj})\,.
\end{equation}
If we totally symmetrize the left- and the right-hand sides of
  (\ref{A=C-S}), an immediate consequence is
\begin{equation}\label{A()=0}
A^{(ijkl)}=0\,.
\end{equation}
If we symmetrize (\ref{firsty+}) with respect to the indices $jkl$,
we recognize that its right-hand side vanishes. Accordingly, we
have the
\begin{prop} The tensor $A$ fulfills the additional symmetry
\begin{equation}\label{A-ident}
A^{i(jkl)}=0\qquad\text{\rm or}\qquad A^{ijkl}+A^{iklj}+A^{iljk}=0\,.
\end{equation}
\end{prop}

\subsubsection{Vector spaces of the  $S$- and $A$-tensors}

We denote the $21$-dimensional vector space of $C$ by ${\cal C}$. The
irreducible decomposition of $C$ signifies the reduction of ${\cal C}$
to the {\it direct sum} of its two subspaces, ${\cal S}\subset{\cal
  C}$ for the tensor $S$, and ${\cal A}\subset{\cal C}$ for the tensor
$A$,
\begin{equation}\label{alg1}
{\cal C}={\cal S}\oplus{\cal A}\,.
\end{equation}
The vector spaces ${\cal S}$ and ${\cal A}$ have only zero in their
intersection and the decomposition of the corresponding tensors is
{\it unique.} According to Proposition 1, the sum of the dimensions of
the subspaces is, equal to 21. The two irreducible parts $S^{ijkl}$
and $A^{ijkl}$ preserve their symmetries under arbitrary linear frame
transformations. In particular, they fulfill the minor and major
symmetries of $C^{ijkl}$ likewise. Accordingly, $S$ and $A$, or $S$
alone (but not $A$ alone, as we will see later) can be elasticity
tensors for a suitable material---in contrast to $M$ and $N$.

The dimensions of the vector spaces of $S$ and $A$ can now be easily
determined.
\begin{prop}
For the vector space ${\cal S}$ of the tensors $S^{ijkl}$,
\begin{equation}\label{15d}
{\rm dim}\,{\cal S}=15\,.
\end{equation}
\end{prop}
\begin{proof}
  The number of independent components of a totally symmetric tensor
  of rank $p$ in n dimensions is
  $\binom{n+p-1}{p}=\binom{n-1+p}{n-1}$ or, for dimension 3 and rank 4,
  $\binom{6}{2}=15$, see Schouten (1954)\cite{SchoutenRicci}.
\end{proof}
According to Proposition 1, we have ${\rm
  dim}\,\mathcal{C}=21$. Because of (\ref{alg1}) and (\ref{15d}), we have
\begin{prop}
For the vector space ${\cal A}$ of the tensors $A^{ijkl}$,
\begin{equation} {\rm dim}\,{\cal A}=6\,.
\end{equation}
\end{prop}

%%%%%%%%%%%%%
\subsubsection{Irreducible parts in Voigt's notation}
%%%%%%%%%%%%%%%

In Voigt's 6d notation we have
\begin{equation}\label{decx1}
C^{IJ}=S^{IJ}+A^{IJ}\,\qquad\text{with}\qquad C^{[IJ]}=S^{[IJ]}=A^{[IJ]}=0\,.
\end{equation}
The 6 $\times$ 6 matrix $S^{IJ}$ has 15 independent components. We
choose the following ones [see Voigt (1928),
Eq.~(36) on p.578],
\begin{eqnarray}\label{sym-voigt}
  &&{}S^{11}=C^{11}\,,\quad {}S^{22}=C^{22}\,,\quad {}S^{33}=C^{33}\,,\nonumber\\
  &&{}S^{15}=C^{15}\,,\quad {}S^{16}=C^{16}\,,\quad {}S^{26}=C^{26}\,,\nonumber\\
  &&{}S^{24}=C^{24}\,,\quad {}S^{34}=C^{34}\,,\quad {}S^{35}=C^{35}\,,\nonumber\\
  &&{}S^{12}=\frac 13 \left(C^{12}+2C^{66}\right),\quad {}S^{13}=
  \frac 13 \left(C^{13}+2C^{55}\right),\nonumber\\
  &&{}S^{14}=\frac 13 \left(C^{14}+2C^{56}\right),\quad {}S^{23}=
  \frac 13 \left(C^{23}+2C^{44}\right),\nonumber\\
  &&{}S^{25}=\frac 13 \left(C^{25}+2C^{46}\right),\quad {}S^{36}=
  \frac 13 \left(C^{36}+2C^{45}\right).%\nonumber\\
\end{eqnarray}
The 6 $\times$ 6 matrix of $A^{IJ}$ has 6 independent components.  We
choose the following ones,
\begin{eqnarray}\label{asym-voigt}
  &&{}A^{12}=\frac 23 \left(C^{12}-C^{66}\right),\quad {}A^{13}
  =\frac 23 \left(C^{13}-C^{55}\right),\quad
  {}A^{14}=\frac 23 \left(C^{14}-C^{56}\right),\nonumber\\
  && {}A^{23}
  =\frac 23 \left(C^{23}-C^{44}\right),\quad{}A^{25}=\frac 23 \left(C^{25}-C^{46}\right),\quad {}A^{36}
  =\frac 23 \left(C^{36}-C^{45}\right)\,.
\end{eqnarray}
The decomposition (\ref{decx1}) can be explicitly presented as
\begin{eqnarray}\label{decx1*}\hspace{-30pt}
\begin{bmatrix}
  C^{11} & C^{12} & C^{13} & C^{14} & C^{15} & C^{16} \\
* & C^{22} & C^{23} & C^{24} & C^{25} & C^{26} \\
* & * & C^{33} & C^{34} & C^{35} & C^{36} \\
* & * & * & C^{44} & C^{45} & C^{46} \\
* & * & * & * & C^{55} & C^{56} \\
* & * & * & * & * & C^{66} \end{bmatrix}&=&
\begin{bmatrix}
  {\bm S^{11}} & {\bm S^{12}} & {\bm S^{13}} & {\bm S^{14}} &
  {\bm S^{15}} &{\bm S^{16}} \\
  * &{\bm S^{22}} & {\bm S^{23}} &{\bm S^{24}} &{\bm  S^{25}} &{\bm S^{26} }\\
  * & * &{\bm S^{33}} &{\bm S^{34}} &{\bm S^{35}} & {\bm S^{36}} \\
  * & * & * & S^{23} & S^{36} & S^{25} \\
  * & * & * & * & S^{13} & S^{14} \\
  * & * & * & * & * & S^{12} \end{bmatrix}\nonumber\\
&+&% \hspace{-90pt}
\begin{bmatrix}
  0 & {\bm A^{12}} &{\bm A^{13}} &{\bm A^{14}} & 0 & 0 \\
  * & 0 &{\bm A^{23}} & 0 &{\bm A^{25} }& 0 \\
  * & * & 0  & 0 & 0 & {\bm A^{36}} \\
  * & * & * & -{\scriptstyle \frac 12} A^{23} &-{\scriptstyle \frac
    12} A^{36} & -{\scriptstyle \frac 12} A^{25} \\
  * & * & * & * & -{\scriptstyle \frac 12}A^{13} &
  -{\scriptstyle \frac 12} A^{14} \\
  * & * & * & * & * & -{\scriptstyle \frac 12} A^{12} \end{bmatrix}\,.
\end{eqnarray}
Here, we use boldface for the independent components of the
tensors. Note that all three matrices are symmetric.

%\newpage
\subsection{Comparing the  $S\!A$- and the $M\!N$-decompositions with each other}

We would now like to compare the two different decompositions:
\begin{eqnarray}\label{decomp1}
  \underbrace{C^{ijkl}}_{21}= \underbrace{S^{ijkl}}_{15}+ \underbrace{A^{ijkl}}_{6}
= \underbrace{M^{ijkl}}_{21}+ \underbrace{N^{ijkl}}_{6}\,.
\end{eqnarray}
The dimensions of the corresponding vector spaces are displayed
explicitly. This makes it immediately clear that $A$ can be expressed
in terms of $N$ and vice versa. Take the antisymmetric part of
(\ref{decomp1}) with respect to $j$ and $k$ and find:
\begin{prop}
The auxiliary quantity $N^{ijkl}$ can be expressed in terms of the
irreducible elasticity $A^{ijkl}$ as follows:
\begin{eqnarray}\label{decomp2}
  N^{ijkl}=A^{i[jk]l}\,.
\end{eqnarray}
Its inverse reads,
\begin{eqnarray}\label{decomp4}
  A^{ijkl}= {\scriptstyle\frac 43} N^{ij(kl)}\,.
\end{eqnarray}
\end{prop}
\begin{proof}
Resolve (\ref{decomp2}) with respect to $A$. For this purpose we recall
that $A$ obeys the right minor symmetry: $A^{ijkl}=A^{ij(kl)}$. This
suggests to take the symmetric part of (\ref{decomp2}) with respect to
$k$ and $l$. Then,
\begin{eqnarray}\label{decomp3}
  N^{ij(kl)}&=&{\scriptstyle \frac
    12}\left(A^{ijkl}-A^{i(kl)j}\right)={\scriptstyle 
    \frac 14} \left(2A^{ijkl}-A^{iklj}-A^{ilkj}\right)\cr
  &\stackrel{(37)}{=}&{\scriptstyle 
    \frac 14} \left[2A^{ijkl}-(-A^{iljk}-A^{ijkl})-A^{ilkj}\right]
  = {\scriptstyle  \frac 34}A^{ijkl} \,.
\end{eqnarray}
\end{proof}

Both, $A$ and $N$ have only 6 independent components. In 3d this means
that it must be possible to represent them as a symmetric tensor of
2nd rank. With the operator $\frac 12\epsilon_{mij}$, we can always
map an antisymmetric index pair $ij$ to a corresponding vector index
$m$. The tensor $A^{ijkl}$ has 4 indices, that is, we have to apply
the $\epsilon$ operator twice. Since $A^{ijkl}$ obeys the left and
right minor symmetries, that is, $A^{[ij]kl}=A^{ij[kl]}=0$, the
$\epsilon$ has always to transvect one index of the first pair and one
index of the second pair. This leads, apart from trivial
rearrangements, to a suitable definition.
{\begin{prop}
The irreducible elasticity $A^{ijkl}$ can be equivalently described by
a symmetric 2nd rank tensor
\begin{equation}\label{Delta}
  \Delta_{mn}:=\frac 14\epsilon_{mik}\epsilon_{njl}A^{ijkl}\,,
\end{equation}
with the inverse
\begin{equation}\label{Delta-1}
N^{ijlk}=\epsilon^{ikm}\epsilon^{jln}\Delta_{mn}\qquad\text{or}
\qquad A^{ijkl}=\frac 43\epsilon^{i(k|m}\epsilon^{j|l)n}\Delta_{mn}\,.
\end{equation}
\end{prop}
\begin{proof}
The symmetry of $ \Delta_{mn}$ can be readily established:
\begin{eqnarray}\label{DeltaSym}
  \Delta_{[mn]}=\frac 18\left(\epsilon_{mik}\epsilon_{njl}A^{ijkl}-
    \epsilon_{nik}\epsilon_{mjl}A^{ijkl}\right)
  =\frac 18\epsilon_{mik}\epsilon_{njl}\!\left(A^{ijkl}-
    A^{jilk}\right)=0\,.\end{eqnarray}
Eq.~(\ref{Delta-1})$_1$ can be derived by substituting
(\ref{Delta}) into its right-hand side and taking care of
(\ref{decomp2}). Eq.~(\ref{Delta-1})$_2$ then follows by applying 
(\ref{decomp4}).
\end{proof}} The symmetric 2nd rank tensor $\Delta_{mn}$ (differing by
a factor 2) was introduced by Hauss\"uhl (1983, 2007), its relation to
the irreducible piece $A$ was found by Hehl \& Itin (2002).
By means of (\ref{A=C-S}), $\Delta_{mn}$ can be calculated directly
from the undecomposed elasticity tensor:
\begin{equation}\label{DeltaC}
  \Delta_{mn}=\frac 14\epsilon_{mik}\epsilon_{njl}\,C^{ijkl}\,.
\end{equation}

The tensor $M$, like $N$, can also be expressed in terms of
irreducible pieces: We substitute (\ref{decomp2}) into
(\ref{decomp1}),
\begin{eqnarray}\label{decomp5}
 C^{ijkl}=S^{ijkl}+A^{ijkl}=M^{ijkl}+A^{i[jk]l}\,,
\end{eqnarray}
and resolve it with respect to $M$:
\begin{eqnarray}\label{decomp6}
  M^{ijkl}&=&S^{ijkl}+A^{ijkl}-A^{i[jk]l}\cr
  &=&S^{ijkl}+ {\scriptstyle  \frac 12}\left(2A^{ijkl}-A^{ijkl}+A^{ikjl} \right)\,.
\end{eqnarray}
{We collect the terms and find
\begin{prop}
  The auxiliary quantity $M^{ijkl}$ can be expressed in terms of the
  irreducible elasticities as follows:
\begin{eqnarray}\label{decomp7}
M^{ijkl}=S^{ijkl}+A^{i(jk)l}\,.
\end{eqnarray}
\end{prop}}
According to Proposition 3, the reducible parts $M$ and $N$ don't obey
the left and the right minor symmetries. Consequently---in contrast to
$S$ and $A$, which both obey the major {\it and} the minor
symmetries---$M$ and $N$ cannot be interpreted as directly observable
elasticity tensors. Therefore, in all physical applications we are
forced to use eventually the irreducible pieces $S$ and $A$. The
reducible parts $M$ and $N$ are no full-fledged substitutes for them
and can at most been used for book keeping.

\section{Physical applications of the irreducible decomposition}

\subsection{{Cauchy relations, two kinds of elasticity}}

{Having the $S\!A$-decomposition at our disposal, it is
  clear that we can now classify elastic materials. The anisotropic
  material with the highest symmetry is that for which the $A$
  elasticity vanishes:
\begin{equation}\label{A=0}
  A^{ijkl}=0\qquad\text{or}\qquad \Delta_{mn}=0\qquad\text{or}\qquad N^{ijkl}=0\,.
\end{equation}
The last equation, in accordance with the definition (\ref{MNdef}) of $N$,
can also be rewritten as
\begin{equation}\label{cauchy1a}
 C^{ijkl}=C^{ikjl}\,.
\end{equation}
These are the so-called Cauchy relations, for their history, see
Todhunter (1960)\cite{Todhunter}. The representation in
(\ref{cauchy1a}) is widely used in elasticity literature, see, for
example, Hauss\"uhl (1983)\cite{Haus1983}, Cowin
(1989)\cite{Cowin1989}, Cowin \& Mehrabadi (1992)\cite{Cowin1992},
Campanella \& Tonon (1994)\cite{Campanella}. Podio-Guidugli
(2000)\cite{Podio-Guidugli}, Weiner (2002)\cite{Weiner}, and Hehl \&
Itin (2002)\cite{Cauchy}.

In Voigt's notation, we can use $A^{IJ}=0$ in (\ref{asym-voigt}) and
find the following form of the Cauchy relations,
\begin{eqnarray}\label{Cauchy1}
  &&C^{12}=C^{66}\,,\qquad\qquad C^{13}=C^{55}\,,\qquad\qquad
  C^{14}=C^{56}\,,\nonumber\\
  && C^{23}=C^{44}\,,\qquad\qquad C^{25}=C^{46}\,,\qquad\qquad C^{36}=C^{45}\,,
\end{eqnarray}
see Love (1927)\cite{Love} and Voigt (1928)\cite{Voigt}. Of course,
the same result can also be read off from (\ref{asym-alter}) for
$N=0$.}

Let us first notice that for most materials the Cauchy relations do
not hold even approximately. In fact, the elasticity of a generic
anisotropic material is described by the whole set of the 21
independent components $C^{IJ}$ and not by a restricted set of 15
independent components obeying the Cauchy relations, see Hauss\"uhl
(2007)\cite{Haus}. This fact seems to nullify the importance of the
Cauchy relations for modern solid state theory and leave them only as
historical artifact.

However, a lattice-theoretical approach to the elastic constants
shows, see Leibfried (1962)\cite{Leibfried}, that the Cauchy relations
are valid provided (i) the interaction forces between the atoms or
molecules of a crystal are central forces, as, for instance, in rock
salt, (ii) each atom or molecule is a center of symmetry, and (iii)
the interaction forces between the building blocks of a crystal can be
well approximated by a harmonic potential{, see also Perrin
  (1979)\cite{Perrin}}. In most elastic bodies this is not fulfilled
at all, see the detailed discussion in Hauss\"uhl
(2007)\cite{Haus}. Accordingly, a study of the {\it violations} of the
Cauchy relations yields important information about the intermolecular
forces of elastic bodies. One should look for the {\it deviation of
  the elasticity tensor from its Cauchy part.} { Recently,
  Elcoro \& Etxebarria (2011)\cite{Elcoro} pointed out that the
  situation is more complex than thought previously, for details we
  refer to their article.}

  This deviation measure, being a macroscopic characteristic of the
  material, delivers important information about the microscopic
  structure of the material. { It must be defined in terms
    of a unique proper decomposition of the elasticity
    tensor. Apparently,} the tensor $N^{ijkl}$ cannot serve as a such
  deviation, {in contrast to the stipulations of
    Podio-Guidugli (2000), for example, because it is not an
    elasticity tensor and} because its co-partner $M^{ijkl}$ has 21
  components, {that is, as many} as the elasticity tensor
  $C^{ijkl}$ itself. Only when $N^{ijkl}=0$ is assumed, the tensor
  $M^{ijkl}$ is restricted to 15 independent components. The problem
  of the identification of the deviation part is solved when the
  irreducible decomposition is { used. In this case, we
    can define the main or Cauchy part, as given by the tensor
    $S^{ijkl}$ with 15 independent components, and the deviation or
    non-Cauchy part, presented by the tensor $A^{ijkl}$ with 6
    independent components.}

We can view the elasticity tensor as being composed of two independent
parts, $S$ and $A$.  Due to the irreducible decomposition the set of
elastic constants can be separated into two subsets which are
components of two independent tensors, $C=S\oplus A$. In this way, the
completely symmetric Cauchy part of the elasticity tensor $S^{ijkl}$
has an independent meaning. The additional part $A^{ijkl}$ can be
referred than as as a non-Cauchy part. Why do pure {\it Cauchy
  materials} { ($S\ne0,\,A=0$)} and pure {\it non-Cauchy
  materials} { ($S=0,\,A\ne0$)} not exist in nature?  Can
such pure types of materials be designed artificially? Or does some
principal fact forbid the existence of pure Cauchy and pure non-Cauchy
materials? These questions seem to be important for elasticity theory
and even for modern material technology.  Note that in the framework
of the reducible $M\!N$-decomposition such questions cannot even be
raised. It is because the reducible $M$ and $N$ parts themselves
cannot serve as independent elasticity tensors.  We will address
subsequently the reason why pure non-Cauchy materials are forbidden.

%%%%%%%%%%%%%%%%%%%%%%%%%%%%%
\subsection{Elastic energy}
%%%%%%%%%%%%%%%%%%%%%%%%%%%%

{
In linear elasticity, using the generalized Hooke law,
  the elastic energy is given by
\begin{equation}\label{energy1} 
  W=\frac 12  \,\sigma^{ij}\varepsilon_{ij}\
  =\frac 12 C^{ijkl}  \ve_{ij}\ve_{kl}\,.
\end{equation}
Because of the irreducible decomposition $C=S+A$, we can split this
energy in a Cauchy part and a non-Cauchy part:
\begin{equation}\label{energy2} 
  W  =\, ^{(\rm C)}\!W+{} ^{\rm  (nC)}\!W\,, \qquad\text{with}\qquad
  ^{(\rm C)}\!W := \frac 12 S^{ijkl} \varepsilon_{ij} \varepsilon_{kl}\,,
  \> ^{\rm (nC)}\!W := \frac 12 A^{ijkl} \varepsilon_{ij} \varepsilon_{kl}\,. 
\end{equation}
As we have seen in the Sec.~3.1, this splitting makes perfectly good
sense in physics, since $S$ as well as $A$, due to their symmetries
and dimensions, are elasticity tensors themselves.

Since the strain $\varepsilon_{ij}$ can be expressed in terms of the
displacement gradients according to $\varepsilon_{ij}
=2\partial_{(i}u_{j)}$, we can do the analogous for the elastic
energy:
\begin{equation}\label{energy3} W=\frac 12 C^{ijkl}
  \ve_{ij}\ve_{kl} =2C^{ijkl}\partial_{(i}u_{j)}\partial_{(k}u_{l)}
  =2C^{ijkl}\partial_{i}u_{j}\partial_{k}u_{l}\,.
\end{equation}
We can drop both pairs of parentheses $_{()}$ because the corresponding
symmetries are imprinted already in the elasticity tensor.

%%%%%%%%%%%%%%%%%%%%%%%%%%%%%%%%%%
\subsection{Null Lagrangians in linear elasticity?}
%%%%%%%%%%%%%%%%%%%%%%%%%%%%%%%%%%

We would like now to discuss a proposal on null Lagrangians by
Podio-Guidugli (2000)\cite{Podio-Guidugli}, see also literature cited
by him. We substitute the decomposition (\ref{MNdef}) into
(\ref{energy3}):
\begin{equation}\label{energy4}
  W=2C^{ijkl}\partial_{i}u_{j}\partial_{k}u_{l}=
  2M^{ijkl}\partial_{i}u_{j}\partial_{k}u_{l}+
  2N^{ijkl}\partial_{i}u_{j}\partial_{k}u_{l}\,.
\end{equation}
We turn our attention now to the last term and to the antisymmetry of
$N$, namely to $N^{i(jk)l}=0$. There is a subtlety involved. We want
to integrate partially in the last term. For exploiting the
antisymmetry of $N$ in $j$ and $k$, we need the partial differentials
$\partial_j$ and $\partial_k$. Therefore, we use the left minor
symmetry of $C$ and rewrite (\ref{energy4}) as
\begin{equation}\label{energy5}
  W=2C^{ijkl}\partial_{j}u_{i}\partial_{k}u_{l}=
  2M^{ijkl}\partial_{j}u_{i}\partial_{k}u_{l}+
  2N^{ijkl}\partial_{j}u_{i}\partial_{k}u_{l}\,.
\end{equation}
Of course, the sum of the two terms has not changed, but each single
term did change since $M$ and $N$ do {\it not} obey the left minor
symmetry.  Now we can partially integrate in the last term:
\begin{equation}\label{energy6*}
 2N^{ijkl}\partial_{j}u_{i}\partial_{k}u_{l}
= 2N^{ijkl}\left[\partial_{j}(
      u_i\partial_k u_l ) -u_i\partial_j\partial_k
    u_l\right]\,.
\end{equation}
Since  $N^{i(jk)l}=0$ and  $\partial_{[j}\partial_{k]}=0$, the last
term drops out and we are left with 
\begin{eqnarray}\label{energy6}
  W=2M^{ijkl}\partial_{j}u_{i}\partial_{k}u_{l}+2N^{ijkl}\partial_{j}(
      u_i\partial_k u_l )\,.
\end{eqnarray}

The} second $N$-term is a total derivative, as was already shown by
Lancia et al.(1995)\cite{Lancia}. Thus, only the $M$-term is involved
in the variational principle for deriving the equations of
motion. This result was interpreted by Lancia et
al.(1995)\cite{Lancia} { as solving the null Lagrangian
  problem for linear elasticity theory, since the additional
  $M$-tensor is the only one that is involved in the equilibrium
  equation. However,} this statement does not have an invariant
meaning. Indeed, the remaining part with the $M$-tensor has exactly
the same set of 21 independent components as the initial $C$-tensor,
see Proposition 5. Moreover, this tensor allows a successive
decomposition which can include additional null Lagrangian terms.

{We find from (\ref{energy6}), using the Propositions
  9 and 11 and subsequently Proposition 6,
\begin{eqnarray}\label{energy7}
 W=\left(2S^{ijkl}-A^{iljk}\right)\partial_{j}u_{i}\partial_{k}u_{l}
+2 A^{i[jk]l}\partial_{j}\left(
    u_i\partial_k u_l  \right)\,.
\end{eqnarray}
Consequently} the $A$-tensor is included in the total derivative term,
which vanishes if the Cauchy relations hold. However the same
$A$-tensor appears together with the $S$-tensor also in the first part
of the energy functional.

For static configurations, (\ref{energy1}) can play a role of a
Lagrangian functional whose variation with respect to the displacement
field generates the equilibrium equation. The null-Lagrangian is
defined as that part of the strain energy functional that does not
contribute to the equilibrium equation. The problem is to identify the
null-Lagrangian part of the elasticity Lagrangian and consequently to
establish which set of the independent components of the elasticity
tensor contributes to the equilibrium equation.

{
As a cross-check for our considerations, we determine the
  equilibrium conditions for the Lagrangian (\ref{energy7}). Up to
a total derivative term, the variation of the Lagrangian (\ref{energy7})
reads
 \begin{equation}\label{lag5}
   \delta W=\left(2S^{ijkl}-A^{iljk}\right)
   \left({\partial_j  u_i}{\partial_k \delta u_l}+{\partial_j \delta
       u_i}{\partial_k  u_l}\right)\,.
\end{equation}
Since the minor and the major symmetries hold for the
$S$ and $A$ tensors, the last two terms can be summed up:
 \begin{equation}\label{lag5x}
   \delta W=2\left(2S^{ijkl}-A^{iljk}\right){\partial_j \delta
       u_i}{\partial_k  u_l}\,.
\end{equation}
We integrate partially,
 \begin{equation}\label{lag6}
   \delta W=2\left(2S^{ijkl}-A^{iljk}\right){\partial_j }\left(\delta u_i
     {\partial_k u_l}\right)-2\left(2S^{ijkl}-A^{iljk}\right) 
   \delta u_i{\partial_j\partial_k u_l}\,,
\end{equation}
and can read off the equilibrium conditions as
\begin{equation}\label{lag7}
  \left(2S^{ijkl}-A^{iljk}\right) {\partial_j\partial_k u_l}=0\,.
\end{equation}

Only at a first glance, this equation seems to be new.
Indeed, we can rewrite it by using the Propositions 6 and
  11,
\begin{equation}\label{lagr6a}
  \left(2S^{ijkl}+A^{ijkl}+A^{iklj}\right)
  {\partial_j\partial_k u_l}
  = 2\left(S^{ijkl}+A^{i(jk)l}\right)
  {\partial_j\partial_k u_l}
= 2M^{ijkl}{\partial_j\partial_k u_l}=0\,,
\end{equation}
Since $M^{ijkl}=C^{i(jk)l}$, we can substitute it and, because of 
$\partial_{[j}\partial_{k]}=0$, we have the  standard equilibrium equation
\begin{equation}\label{lagr7a}
C^{ijkl} {\partial_j\partial_k u_l}=0\,.
\end{equation}}
Our calculations confirms the following:

\begin{prop} {\it Any total derivative term in the elastic energy
    functional can be regarded as only a formal expression. It does
    not remove any subset of the elastic constants from the
    equilibrium equation, that is, for an arbitrary anisotropic
    material an elastic null Lagrangian does not exist. }
\end{prop}

%%%%%%%%%%%%%%%%%%%%%%%%%%%%%
\subsection{Wave equation}
%%%%%%%%%%%%%%%%%%%%%%%%%%%%

{The wave propagation in linear elasticity for} anisotropic
  media is described by the following equation:
\begin{equation}\label{wave-eq}
{  \rho g^{il}\ddot{u_l}-C^{ijkl}\,{\partial_j\partial_k u_l}=0\,.}
\end{equation}
Here the { displacement} covector $u_l$ is assumed to be a
function of the time coordinate and of the position of a point in the
{medium.} All other coefficients, the mass density $\rho$,
the elasticity tensor $C^{ijkl}$, and the metric tensor $g^{il}$, are
assumed to be constant; {moreover, we use Cartesian
  coordinates, that is, the Euclidean metric reads
  $g^{ij}=\text{diag}(1,1,1)$.}

We {make} a {plane} wave ansatz, with the
notation as in Nayfeh (1985)\cite{Nayfeh}:
\begin{equation}\label{wave-an}
u_i=U_ie^{i\left(\zeta n_jx^j-\omega t\right)}\,.
\end{equation}
Here $U_i$ is {the covector of} a complex constant
amplitude, $\zeta$ the wave-number, $n_j$ the propagation unit
covector, $\omega$ the {angular} frequency, and $i^2=-1$.
Substituting (\ref{wave-an}) into (\ref{wave-eq}), we obtain a system
of three homogeneous {algebraic} equations
\begin{equation}\label{wave-an1}
\left(\rho\, \omega^2g^{il}-C^{ijkl}\zeta^2n_jn_k\right)U_l=0\,.
\end{equation}
It has a non-trivial solution if and only if
 the characteristic equation holds,
\begin{equation}\label{char}
\det\left(\rho \omega^2g^{il}-C^{ijkl}\zeta^2n_jn_k\right)=0\,.
\end{equation}

{With the definitions of} the Christoffel tensor
\begin{equation}\label{Christoffel}
\Gamma^{il}:=\frac 1\rho\, C^{ijkl}n_jn_k\,
\end{equation}
and of the phase velocity $v:=\omega/\zeta$, the system
(\ref{wave-an1}) takes the form
\begin{equation}\label{char1}
\left( v^2g^{il}-\Gamma^{il}\right)U_l=0\,,
\end{equation}
{whereas} the characteristic equation  reads
\begin{equation}\label{char2}
\det\left( v^2g^{il}-\Gamma^{il}\right)=0\,.
\end{equation}
Due to the {minor and major} symmetries of the elasticity
tensor, the Christoffel tensor {turns out to be}
symmetric
\begin{equation}\label{Christoffel1}
{\Gamma^{[ij]}=0\,.}
\end{equation}

For {a} symmetric matrix, the characteristic equation has
only real solutions.  Every positive real solution corresponds to an
acoustic wave propagating in the direction of the wave covector $n_i$.
Thus, in general, for a given propagation covector $n_i$, three
different acoustic waves {are possible}. The cases with
zero solutions (null modes) and negative solutions (standing waves)
must be treated as unphysical because they do not satisfy the
causality requirements.

A symmetric tensor, by itself, cannot {be} decomposed
{directly} under the action of the group $GL(3,\mathbb R)$. However,
the irreducible decomposition of the elasticity tensor generates the
corresponding {\it decomposition of the Christoffel tensor}.
Substituting the irreducible $S\!A$-decomposition into
(\ref{Christoffel}), we obtain
\begin{equation}\label{Chr-dec1}
\Gamma^{il}=\S^{il}+\A^{il}\,,
\end{equation}
where
\begin{equation}\label{Chr-dec2}
  \S^{il}:=S^{ijkl}n_jn_k{=\S^{li}}\,,\qquad {\rm and }\qquad
  \A^{il}=A^{ijkl}n_jn_k{=   \A^{li} }                \,.
\end{equation}
These two symmetric tensors correspond to the Cauchy and non-Cauchy
parts of the elasticity tensor. We will call $\S^{ij}$ {\it Cauchy
  Christoffel tensor} and $\A^{ij}$ {\it non-Cauchy Christoffel
  tensor.}  {Substituting} (\ref{firsty}) and
(\ref{firsty+}) into (\ref{Chr-dec2}), we find the explicit
expressions
\begin{equation}\label{Chr-dec3}
\S^{il}=\frac 1{3\rho}(C^{ijkl}+C^{iklj}+ C^{iljk})n_jn_k\,
\end{equation}
and
\begin{equation}\label{Chr-dec4}
\A^{il}=\frac1{3\rho}\left(2C^{ijkl}-C^{iklj}-C^{ilkj}\right)n_jn_k\,.
\end{equation}
{ If we transvect $n_l$ with $\A^{il}$, we} observe a
generic fact:
\begin{prop}
  For {each} elasticity tensor $C^{ijkl}$ and for
  {each} wave covector $n_i$,
\begin{equation}\label{prop1}
\A^{il}n_l=0\,.
\end{equation}
\end{prop}
\begin{proof}
  {We have
    $\A^{il}n_l=A^{ijkl}n_jn_kn_l=A^{i(jkl)}n_jn_kn_l
\stackrel{(37)}{=}0$. }
\end{proof}
\begin{prop}
The determinant of the non-Cauchy  Christoffel tensor is equal to zero,
\begin{equation}\label{detA}
\det(\A^{ij})=0\,.
\end{equation}
\end{prop}
\begin{proof}
  Equation (\ref{prop1}) can be viewed as a linear relation between
  the {rows} of the matrix $\A^{il}$; {this}
  proves its singularity.
\end{proof}

The wave equation (\ref{char1}) {can be} rewritten as
\begin{equation}\label{char1x}
\left( v^2g^{il}-\S^{il}-\A^{il}\right)U_l=0\,,
\end{equation}
with the characteristic equation
\begin{equation}\label{char3x}
\det\left( v^2g^{il}-\S^{il}-\A^{il}\right)=0\,.
\end{equation}
We can now recognize  {the} reason why a pure non-Cauchy
{medium} is forbidden. In the case $\S^{il}=0$, the
characteristic equation (\ref{char3x}) takes the form
\begin{equation}\label{char4x}
\det\left( v^2g^{il}-\A^{il}\right)=0\,.
\end{equation}
Since $\det (\A^{il})=0$, at least one of its eigenvalues is
zero. Such a null-mode wave is forbidden {because of}
causality {reasons}.

The Christoffel tensor is real and symmetric, thus all its eigenvalues
are real and the associated eigenvectors are orthogonal.  In order to
have three real positive eigenvalues, we need to satisfy the condition
of positive definiteness of the matrix $\Gamma^{ij}$.  In this case,
the following possibilities arise: 
{\begin{itemize}
  \item[{ {(i)}}] All the eigenvalues are distinct
    \hspace{20pt}$v_1^2>v_2^2>v_3^2\,,$
  \item[{(ii)}] two eigenvalues coincide
    \hspace{55pt}$v_1^2>v_2^2=v_3^2\>\;$ or $\>v_1^2=v_2^2>v_3^2\,,\;$
    or
  \item[{(iii)}] three eigenvalues coincide
    \hspace{48pt}$v_1^2=v_2^2=v_3^2\,.$
\end{itemize}}
The Christoffel matrix depends on the propagation vector
$\Gamma^{ij}({\vec{n}})$, thus the conditions indicated above
determine the directions in which the wave can propagate.

%%%%%%%%%%%%%%%%%%%%%%%%%%%
\subsection{Polarizations}
%%%%%%%%%%%%%%%%%%%%%%%%%%%%

Equation (\ref{char1}), or its decomposed form (\ref{char1x}),
represents acoustic wave propagation in an elastic
{medium}. It is an eigenvector problem {in}
which eigenvalues $v^2$ are the solutions of (\ref{char3x}). In
general, three distinct real positive solutions correspond to three
independent waves $^{(1)}U_l$, $^{(2)}U_l$, and $^{(3)}U_l$, called
acoustic polarizations, see Nayfeh (1985)\cite{Nayfeh}.

For isotropic materials, there are three pure polarizations: one {\it
  longitudinal {or compression} wave} with
\begin{equation}\label{long}
 \vec{U}\times\vec{n}=0\,,
\end{equation}
that is, the polarization is directed along the propagation vector,
and two {\it {transverse or} shear waves with}
\begin{equation}\label{shear}
	  \vec{U}\cdot\vec{n}=0\,,
\end{equation}
that is, the polarization is normal to the the propagation vector. In
general, for anisotropic materials, three pure modes do not exist. The
identification of the pure modes and the condition for their existence
is an interesting problem.

Let us see how the irreducible decomposition of the elasticity tensor,
which we applied to the  Christoffel tensor, can be used here. For a
chosen direction vector $\vec{n}$, we {introduce a vector
  and a scalar according to}
\begin{equation}\label{vec-scal}
  S^i:=\Gamma^{ij}n_j\,,\qquad S:=\Gamma^{ij}n_in_j\,.
\end{equation}
{This notation is consistent since $S^i$ and $S$,} due to 
(\ref{prop1}), depend only on the Cauchy part of the elasticity tensor, 
\begin{equation}\label{vec-scal1}
S^i=S^{ij}n_j\,,\qquad S=S^{ij}n_in_j\,.
\end{equation}
\begin{prop}
  Let $n^i$ denotes {an allowed direction for the
    propagation of a compression wave.} Then the velocity
  $v_{\rm{L}}$ of { this wave} in the direction of $n^i$
    is determined only by the Cauchy part of the elasticity tensor:
\begin{equation}\label{long-vel}
v_{\rm{L}}=\sqrt S\,.
\end{equation}
\end{prop}
\begin{proof}
For the longitudinal wave, $u_j=\alpha n_j$. Thus, (\ref{char1})
 {becomes} 
\begin{equation}\label{char1-x}
  \left( v^2g^{ij}-\Gamma^{ij}\right)n_j=0\qquad\text{or}\qquad
  v^2g^{ij}n_j=S^{ij}n_j\,.
\end{equation}
{Transvecting both sides of the last equation with} $n_i$,
we obtain (\ref{long-vel}).
\end{proof}
Let us now discuss in which directions the three pure polarizations
can propagate.
 \begin{prop}
   For a {medium} with a given elasticity tensor, all
   three pure{ly}  polarized waves (one longitudinal and two
   {transverse}) can propagate in the direction
   {${\vec{n}}$} if and only if
\begin{equation}\label{long-cond}
S^i= Sn^i\,.
\end{equation}
\end{prop}
\begin{proof}
  Since the Christoffel matrix is symmetric, it has real eigenvalues
  and three orthogonal eigenvectors. We have three pure polarizations
  if and only if one of {these} eigenvectors
  {points} in the direction {of
    ${\vec{n}}$}. Consequently,
\begin{equation}\label{long-cond1}
\Gamma^{ij}n_j= v^2_{\text{L}} n^i\,.
\end{equation}
Since $\Gamma^{ij}n_j=S^{ij}n_j=S^i$, {we have
  $S^i=v^2_{\text{L}} n^i$ or, after} substituting (\ref{long-vel}),
we obtain (\ref{long-cond}).
\end{proof} {Accordingly,} for a given
{medium}, the directions of the pure{ly}
polarized waves depend on the Cauchy part of the elasticity tensor
{alone}. In other words, two materials, with the same
Cauchy parts $S^{ijkl}$ {of the elasticity tensor} but
different non-Cauchy parts $A^{ijkl}$, have the same pure wave
propagation directions and the same longitudinal velocity.

\section{Examples}
%%%%%%%%%%%%%%%%%%%%%%%%%%%
\subsection{Isotropic media}
%%%%%%%%%%%%%%%%%%%%%%%%%%%%

In order to clarify the results discussed above, consider as an
example an isotropic elastic medium. Then, the elasticity tensor can
be expressed in terms of the metric tensor $g^{ij}$ as
\begin{equation}\label{iso}
  C^{ijkl}=\lambda\,g^{ij}g^{kl}+\mu\left(g^{ik}g^{lj}+g^{il}g^{jk}
  \right)=\lambda\,g^{ij}g^{kl} +2\mu\,g^{i(k}g^{l)j}\,,
\end{equation}
with the Lam\'e moduli $\lambda$ and $\mu$, see Marsden \& Hughes
(1983)\cite{Marsden}.  The irreducible decomposition of this
expression involves the terms
\begin{equation}\label{first''}
  S^{ijkl}= (\lambda+2\mu)\,g^{(ij}g^{kl)}=
  \frac{\lambda+2\mu}{3}\left(g^{ij}g^{kl}+
    g^{ik}g^{lj}+g^{il}g^{jk}\right)\,,
\end{equation}
and
\begin{equation}\label{second'}
  A^{ijkl}= \frac{\lambda-\mu}{3}\left(2g^{ij}g^{kl}-
     g^{ik}g^{lj}-g^{il}g^{jk}\right)\,.
\end{equation}
The symmetric second rank tensor $\Delta$ of equation (\ref{Delta}),
which is equivalent to $A$, reads
\begin{equation}\label{Delta+}
  \Delta_{mn}=\frac 14\epsilon_{mik}\epsilon_{njl}A^{ijkl}
= \frac{\lambda-\mu}{12}\epsilon_{mik}\epsilon_{njl} \left(2g^{ij}g^{kl}-
     g^{ik}g^{lj}-g^{il}g^{jk}\right)\,
\end{equation}
or
\begin{equation}\label{piso*8}
 \Delta_{ij}= \frac{\lambda-\mu}{2}\,g_{ij}\,.
\end{equation}
Of course, (\ref{piso*8}) is much more compact than (\ref{second'}).

Accordingly, the irreducible decomposition reveals two fundamental
combinations of the isotropic elasticity constants. We denote them by
\begin{equation}\label{irr-cons}
\a:=\frac{\lambda+2\mu}3\,,\quad \b:=\frac{\lambda-\mu}3\,.
\end{equation}
Then we can rewrite the canonical representation (\ref{iso}) of the elasticity
tensor as
\begin{eqnarray}\label{is-can}
  C^{ijkl}&=&\a\left(g^{ij}g^{kl}+g^{ik}g^{lj}+g^{il}g^{jk}\right)+\b\left(2g^{ij}g^{kl}-
    g^{ik}g^{lj}-g^{il}g^{jk}\right)\nonumber\\
  &=&(\a+2\b)\,g^{ij}g^{kl} +2(\a-\b)\,g^{i(k}g^{l)j}\,.
\end{eqnarray}

The alternative $M\!N$-decomposition can also be derived
straightforwardly. We can read off from (\ref{iso}) directly that 
\begin{equation}\label{inver2}
  M^{ijkl}=C^{i(jk)l}=(\lambda+\mu)g^{i(j}g^{k)l}+\mu
  g^{il}g^{jk}\,.
\end{equation}
In a similar way we can find for $N^{ijkl}$ the expression
\begin{equation}\label{inver1}
N^{ijkl}=C^{i[jk]l}=(\lambda-\mu)\,g^{i[j}g^{k]l}\,.
\end{equation}
The key fact is here that this decomposition is characterized by {\it
  three parameters}
\begin{equation}\label{red-cons}
\a'=\frac {\lambda+\mu}2\,,\quad \b'=\mu\,,\quad \g'=\frac {\lambda-\mu}2\,.
\end{equation}
In both  decompositions,  the Cauchy relations are described by the equation
 \begin{equation}\label{iso-cauchy}
 \lambda=\mu\,.
\end{equation}
Only in the special case when the Cauchy relations hold, the tensors
$S^{ijkl}$ and $M^{ijkl}$ are equal.

Consider now the energy functional and the problem of the
identification of the null Lagrangian. When the $S\!A$-decomposition
is substituted into the energy functional, we remain with two
irreducible parts with their leading coefficients $\a$ and $\b$.
Thus, the number of parameters is not reduced.

As for the $M\!N$-decomposition, the $N$ term, with its leading
coefficient $\gamma'$, being a total derivative, does not contribute
to the equilibrium equation. We still remain with two terms with their
leading parameters $\a'$ and $\b'$. Moreover, as it was shown above,
these terms can be recovered in the initial functional. It proves that
there is no such thing as a null Lagrangian in elasticity, not even
in the simplest isotropic case.

In this isotropic case, the main difference between the $S\!A$- and
the $M\!N$-decomposition becomes manifest. The irreducible
$S\!A$-decomposition dictates the existence of two independent
fundamental parameters of the isotropic medium, whereas the reducible
$M\!N$-decomposition deals with three {\it linearly dependent}
parameters.

Let us calculate now the characteristic velocities of the acoustic waves.
Two two Christoffel matrices are
\begin{equation}\label{iso-chris}
{\cal S}^{ij}=\a\left(g^{ij}+2n^in^j\right)\,.
\end{equation}
and
\begin{equation}\label{iso-chris1}
{\cal A}^{ij}=\b\left(g^{ij}-n^in^j\right)\,.
\end{equation}
Hence, 
\begin{equation}\label{iso-chris2}
S^i=3\a n^i\,,\qquad S=3\a\,.
\end{equation}
Observe that the condition (\ref{long-cond}) of pure polarization is
satisfied now identically. Thus, we recover the well known fact that
in isotropic media every direction allows propagation of purely
polarized waves.

In terms of the moduli $\a$ and $\b$, the characteristic equation for
the acoustic waves takes the form
\begin{equation}\label{iso-char}
\det\left[(v^2-\a+\b)g^{ij}-(2\a+\b)n^in^j\right]=0\,.
\end{equation}
The longitudinal wave velocity 
\begin{equation}\label{iso-sol2}
v_1^2=S=3\a=\lambda-2\mu\,,
\end{equation}
is one solution for this equation. Indeed, in this case, the
characteristic equation (\ref{iso-char}) turns into the identity
$\det(g^{ij}-n^in^j)=0$, which proves that (\ref{iso-sol2}) is an
eigenvalue.

We immediately find additional solutions for this equation. For
\begin{equation}\label{iso-sol1}
v_2^2=\a-\b=\mu\,,
\end{equation}
(\ref{iso-char}) turns into the identity $\det(n^in^j)=0$. Moreover,
the rank of this matrix is equal to one: ${\rm
  rank}(n^in^j)=1$. Consequently, the multiplicity of the root
(\ref{iso-sol1}) is equal to two and the third eigenvalue is the same
one:
 \begin{equation}\label{iso-sol1*}
v_3^2=\a-\b=\mu\,.
\end{equation}

With the use of the irreducible decomposition, we can answer now the
following question: {\it Does there exist an anisotropic medium in
  which purely polarized waves can propagate in every direction?}  It
is clear that media with the same Cauchy part of the elasticity tensor
(and arbitrary non-Cauchy part) have the same property in this
respect. As a consequence we have

\begin{prop}
  The most general type of an anisotropic medium that allows
  propagation of purely polarized waves in an arbitrary direction has
  an elasticity tensor of the form
\begin{equation}\label{iso-gen}
% \scriptsize{
C^{ijkl}=
\begin{bmatrix}
  \a & {\a/3+ 2\rho_1} &{\a/3+ 2\rho_2} &2\rho_3 & 0 & 0 \\
* & \a &{\a/3+ 2\rho_4} & 0 &2\rho_5& 0 \\
* & * & \a  & 0 & 0 & 2\rho_6 \\
* & * & * & \a/3-\rho_4 &-\rho_6 & -\rho_5 \\
* & * & * & * & \a/3-\rho_2 & -\rho_3 \\
* & * & * & * & * &\a/3 -\rho_1 \end{bmatrix}\,,
\end{equation}
where $\rho_1,\cdots,\rho_6$ are arbitrary parameters. The velocity
of the longitudinal waves in this medium is
$v_{\rm {L}}=\sqrt{3\alpha} =\sqrt{ \lambda-2\mu} $.
\end{prop}

%%%%%%%%%%%%%%%%%%%%%%%%%%%
\subsection{Cubic media}
%%%%%%%%%%%%%%%%%%%%%%%%%%%%
Cubic crystals are described by three independent elasticity
constants. In a properly chosen coordinate system, they can be put,
see Nayfeh (1985)\cite{Nayfeh}, into the following Voigt matrix:
\begin{equation}\label{cub-voigt}
%\scriptsize{
\begin{bmatrix}
C^{1111} & C^{1122} & C^{1133} & C^{1123} & C^{1131} & C^{1112} \\
* & C^{2222} & C^{2233} & C^{2223} & C^{2231} & C^{2212} \\
* & * & C^{3333} & C^{3323} & C^{3331} & C^{3312} \\
* & * & * & C^{2323} & C^{2331} & C^{2312} \\
* & * & * & * & C^{3131} & C^{3112} \\
* & * & * & * & * & C^{1212}
     \end{bmatrix} \equiv \begin{bmatrix}
  C^{11} & C^{12} & C^{12} & 0 & 0 & 0 \\
* & C^{11} & C^{12} & 0 & 0 & 0 \\
* & * & C^{11} & 0 & 0 & 0 \\
* & * & * & C^{66} & 0 & 0 \\
* & * & * & * & C^{66} & 0 \\
* & * & * & * & * & C^{66} \end{bmatrix}\,.%}
\end{equation}
We decompose it irreducibly by using (\ref{sym-voigt}) and
(\ref{asym-voigt}) and find the Cauchy part
\begin{equation}\label{cub-sym}
% \scriptsize{
S^{ijkl}=\begin{bmatrix}
\tilde{\a} & \tilde{\b} & \tilde{\b} & 0 & 0 & 0 \\
* & \tilde{\a} & \tilde{\b} &  0 & 0 & 0 \\
* & * & \tilde{\a} & 0 & 0 & 0 \\
* & * & * & \tilde{\b} & 0 & 0 \\
* & * & * & * & \tilde{\b} & 0 \\
* & * & * & * & * & \tilde{\b}
     \end{bmatrix} \,,\qquad \tilde{\a}=C^{11},\quad \tilde{\b}=\frac 13 \left(C^{12}+2C^{66}\right)\,,
\end{equation}
and the non-Cauchy part
\begin{equation}\label{cub-ant}
% \scriptsize{
  A^{ijkl}= \begin{bmatrix}
    0 & 2\tilde{\g} & 2\tilde{\g} & 0 & 0 & 0 \\
    * & 0 & 2\tilde{\g} & 0 & 0 & 0 \\
    * & * & 0 & 0 & 0 & 0 \\
    * & * & * & -\tilde{\g} & 0 & 0 \\
    * & * & * & * & -\tilde{\g} & 0 \\
    * & * & * & * & * &-\tilde{\g} \end{bmatrix} \,,
  \qquad \tilde{\g}=\frac 13 \left(C^{12}-C^{66}\right)\,. %}
\end{equation}
Accordingly, the elasticity tensor is expressed in terms of the three
new elastic constants $\tilde{\a}$, $\tilde{\b}$, and
$\tilde{\g}$. For the Cauchy relation we have
\begin{equation}\label{cub-Cauchy}
\tilde{\g}=0\,,\qquad{\rm{or}}\qquad C^{12}=C^{66}\,.
\end{equation}

The Cauchy part of the Christoffel tenor takes the form
\begin{equation}\label{cub-Cau-Christof}
% \scriptsize{
  S^{il}=\begin{bmatrix}
    \tilde{\a} n_1^2+\tilde{\b}(n_2^2+n_3^2) & 2\tilde{\b} n_1n_2 
    & 2\tilde{\b} n_1n_3  \\
    2\tilde{\b} n_1n_2 &\tilde{\a} n_2^2+\tilde{\b}(n_1^2+n_3^2) 
    & 2\tilde{\b} n_2n_3  \\
    2\tilde{\b} n_1n_3  & 2\tilde{\b} n_2n_3  & \tilde{\a} n_3^2
    +\tilde{\b}(n_1^2+n_2^2) 
     \end{bmatrix} \,.
\end{equation}
The corresponding vector and scalar invariants (\ref{vec-scal1}) read
\begin{equation}\label{vec-scal1*}
  S^i=\begin{bmatrix}
    n_1\big(\tilde{\a} n_1^2+3\tilde{\b}(n_2^2+n_3^2)\big) \\ \\
    n_2\big(\tilde{\a} n_2^2+3\tilde{\b}(n_1^2+n_3^2)\big)  \\ \\
    n_3\big(\tilde{\a} n_3^2+3\tilde{\b}(n_1^2+n_2^2)\big)
  \end{bmatrix}=\begin{bmatrix}
    n_1\big((\tilde{\a}-3\tilde{\b}) n_1^2+3\b\big) \\ \\
    n_2\big((\tilde{\a}-3\tilde{\b}) n_2^2+3\tilde{\b}\big)  \\ \\
    n_3\big((\tilde{\a}-3\tilde{\b}) n_3^2+3\tilde{\b}\big)
     \end{bmatrix}
\end{equation}
and
\begin{equation}\label{scal}
  S=(\tilde{\a}-3\tilde{\b}) (n_1^4+n_2^4+n_3^4)+3\tilde{\b}\,,
\end{equation}
respectively. The non-Cauchy part of the Christoffel tensor turns out
to be
\begin{equation}\label{cub-nonCau-Christof}
  A^{il}=\tilde{\g}
  \begin{bmatrix}
    -n_2^2-n_3^2 & n_1n_2 & n_1n_3  \\
    n_1n_2 &-n_1^2-n_3^2   &  n_2n_3  \\
    n_1n_3 & n_2n_3  &-n_1^2-n_2^2 
  \end{bmatrix} \,.
\end{equation}
The corresponding invariants $A^{ij}n_j$ and $A^{ij}n_in_j$ are zero.

Let us use the pure polarization condition $S^i=Sn^i$ in order to
derive the purely polarized propagation directions. We can easily derive
all solutions of this equation
\begin{itemize}
\item{Edges:}
$n_1=1\,, n_2=n_3=0\,, {\rm etc.}$
The longitudinal velocity  is 
 \begin{equation}\label{edges-vel}
v_{\text{L}}=\sqrt{\tilde{\a}}=\sqrt{C^{11}}\,.
\end{equation}
\item{Face diagonals:} $n_1= n_2= \frac 12 {\sqrt{2}}\,, n_3=0\,, {\rm
    etc.}$ The longitudinal velocity is
 \begin{equation}\label{edges-vel1}
   v_{\text{L}}=\frac 12 \sqrt{{2(\tilde{\a}+3\tilde{\b})}}
=\frac 12 \sqrt{2({C^{11}+C^{12}+2C^{66}})}\,.
\end{equation}

\item{Space diagonals:} $n_1= n_2=n_3=\frac 13 {\sqrt{3}}\,, {\rm etc.}$
  The longitudinal velocity is
 \begin{equation}\label{edges-vel1*}
v_{\text{L}}=\frac 13 \sqrt{3( {\tilde{\a}+6\tilde{\b}})}
=\frac 13 \sqrt{3({C^{11}+2C^{12}+4C^{66}})}\,.
\end{equation}
\end{itemize}
Due to Propositions 15 and 16, a wide class of materials with the same
$S$-tensor and an arbitrary $A$-tensor will have exactly the same
directions and the same velocities of the longitudinal waves.  The
elasticity tensor for such materials can be written as
 \begin{equation}\label{cub-gen}
% \scriptsize{
C^{ijkl}=
\begin{bmatrix}
  \tilde{\a} & {\tilde{\b}+ 2\rho_1} &{\tilde{\b}+ 2\rho_2} &2\rho_3 & 0 & 0 \\
* & \tilde{\a} &{\tilde{\b}+ 2\rho_4} & 0 &2\rho_5& 0 \\
* & * & \tilde{\a}  & 0 & 0 & 2\rho_6 \\
* & * & * & \tilde{\b}-\rho_4 &-\rho_6 & -\rho_5 \\
* & * & * & * & \tilde{\b}-\rho_2 & -\rho_3 \\
* & * & * & * & * &\tilde{\b} -\rho_1 \end{bmatrix}\,.
\end{equation}

\section*{Acknowledgments}
This work was supported by the German-Israeli Foundation for
Scientific Research and Development (GIF), Research Grant No.\
1078-107.14/2009. Y.I.\ would like to thank V.I.~Alshits (Moscow) for
helpful discussion.

\appendix
\section{Irreducible decomposition of tensors of rank  $p$}

In order to understand what should be considered as the proper
decomposition of a tensor, which is of rank $p$ in its covariant or
contravariant indices, that is, $T^{ijkl}$ or $T_{ijkl}$, we must look
closer on its precise algebraic meaning. Let us start with a vector
space $V$ over the number field $F$. This construction comes together
with the group of general linear transformations $GL(\dim\!
V,F)$. This group includes all invertible square matrices of the size
$\dim V\times\dim V$ with the entries in $F$. In our case, $F={\mathbb
  R}$ and $V={\mathbb R}^3$. Consequently, we are dealing with the
group $GL(3,\mathbb R)$ that can be considered as a group of
transformations between two bases $\{e_i\}$ and $\{e_{i'}\}$ of $V$:
\begin{equation}\label{basis}
e_{i'}=L_{i'}{}^ie_i\qquad L_{i'}{}^i\in GL(3,\mathbb R)\,.
\end{equation}
The vector space $V$ cannot be decomposed invariantly into a sum of
subspaces. Indeed, every two nonzero vectors of $V$ can be transformed
each other by the use of certain matrix of $GL(3,\mathbb R)$. Thus,
the group $GL(3,\mathbb R)$ acts transitively on $V$.
 
A tensor of rank $p$ is defined as a multi-linear map from the
Cartesian product of $p$ copies of $V$ into the field $F$,
\begin{equation}\label{tensor}
T:\underbrace{V\times \cdots\times V}_{p}\to F\,.
\end{equation}
The set of all tensors $T$ of the rank $p$ compose a vector space by
itself, say ${\mathcal T}$. The dimension of this {\it tensor space}
is equal to $n^p$. Thus, for rank 4 in 3d we have $3^4=81$.  As a
basis in ${\mathcal T}$, we can take tensor products of basis elements
in $V$,
\begin{equation}\label{tensor-bas}
e_{i_1}\otimes e_{i_2}\otimes\cdots\otimes e_{i_p}\,.
\end{equation}
Accordingly, an arbitrary contravariant tensor of rank $p$ can be
expressed as
\begin{equation}\label{tensor-ex}
T=T^{{i_1} {i_2} \cdots {i_p}}\,e_{i_1}\otimes e_{i_2}\otimes\cdots\otimes e_{i_p}\,.
\end{equation}
Under a transformation (\ref{basis}) of a basis of the vector space,
the basis of the tensor space is multiplied by a product of the
matrices $L_{i'}{}^i$. This is a `derived transformation' in the sense
of Littlewood (1944)\cite{little}. An important fact is that the
tensor space ${\mathcal T}$ is decomposed to a direct sum of subspaces
which are invariant under general linear transformations.  For
instance, the span of the basis
\begin{equation}\label{tensor-bas1x}
  e_{(i_1}\otimes e_{i_2)}\otimes\cdots\otimes e_{i_p}\, 
% \frac 12\left(e_{i_1}\otimes e_{i_2}+e_{i_2}\otimes 
%    e_{i_1}\right)\otimes\cdots\otimes e_{i_p}\,
\end{equation}
describes the subspace of all tensors which are symmetric under the
permutation of the two first indices.  Another subspace is obtained as
the span of all basis tensors of the form
\begin{equation}\label{tensor-bas2x}
  e_{[i_1}\otimes e_{i_2]}  \otimes\cdots\otimes e_{i_p}\,,
%  \frac 12\left(e_{i_1}\otimes e_{i_2}-e_{i_2}\otimes e_{i_1}\right)
%  \otimes\cdots\otimes e_{i_p}\,,
\end{equation}
which is the subspace of all tensors antisymmetric under the
permutation of two first indices.

By taking different permutations of the basis vectors we can obtain
different subtensors (elements of the subspace of ${\mathcal T}$) of a
tensor of arbitrary rank. There are all together $p!$ permutations for
the system of $p$ objects.  In fact, {\it every $GL(3,\mathbb R)$-invariant
  subtensor} can be obtained by the operator which permutes the
indices, see Littlewood (1944)\cite{little}.  Since we need the sums of
all the tensors with permuted indices, we must extend the group of
permutations to its group algebra (Frobenius algebra). This
construction involves formal sums, in an addition to the group
multiplication.

In this way, each tensor of rank two or greater can be {\it
  irreducibly decomposed} under the action of the linear group
$GL(3,\mathbb R)$. The original tensor is written as a linear
combination of simpler tensors of rank $p$, which, under the action of
$GL(3,\mathbb R)$, transform only under themselves.  These partial
tensors obey their specific symmetries {\it in addition} to the
symmetries of the initial tensor. Possible types of the irreducible
tensors are determined by the use of Young's tableaux.

\subsection{Example: Young's decomposition of second rank tensors}

Consider a generic second rank tensor $T_{ij}$ which
is decomposed into a sum of its symmetric and antisymmetric parts,
\begin{equation}\label{irr}
  T^{ij}=T^{(ij)}+T^{[ij]}\,.
\end{equation} 
It means that one starts with the tensor space 
\begin{equation}\label{tensor-bas0}
{\mathcal T}={\rm Span}\left\{
e_{i}\otimes e_{j}
\right\}
\,
\end{equation}
The span of the linear combination of basis tensors  
\begin{equation}\label{tensor-bas1}
{\mathcal S}={\rm Span}\{(e_{(i}\otimes e_{j)}\}\,
\end{equation}
composes the subspace of tensors symmetric under the permutation of
two indices.  Indeed, an arbitrary tensor in $ {\mathcal S}$ is
decomposed as
\begin{equation}\label{tensor-bas1-1}
S=S^{ij}e_{(i}\otimes e_{j)}\,.
\end{equation}
Thus we have $S^{[ij]}=0$.  Another subspace ${\mathcal A}$ is
obtained as a span of the linear combination
\begin{equation}\label{tensor-bas2}
{\mathcal A}={\rm Span}\{e_{[i}\otimes e_{j]}\}\,.
\end{equation}
This is a subspace of antisymmetric tensors,
 \begin{equation}\label{tensor-bas1-2}
A=A^{ij}e_{[i}\otimes e_{j]}\,.
\end{equation}
Here $A^{(ij)}=0$.  In this way, the tensor space is represented as a
direct sum of its subspaces ${\mathcal T}={\mathcal S}\oplus{\mathcal
  A}$. These subspaces are invariant under the transformations of
$GL(3,{\mathbb R})$ while further decomposition into smaller subspaces
is impossible. Hence we have an irreducible decomposition.

In Young's description, this decomposition is given by two diagrams
that are graphical representations of the permutation group $S_2$
\begin{equation}\label{irr1}
\Yvcentermath1\yng(1)\otimes\yng(1)=\yng(2)\oplus \yng(1,1)\,.
\end{equation}
The left-hand side here denotes the tensor product of two vectors,
i.e., a generic asymmetric second order tensor. The right-hand side is
given as a sum of two second order tensors. The symmetric tensor is
represented by the row diagram while the antisymmetric tensor is
given by the column diagram. In our simplest example, only two tables
given in (\ref{irr1}) are allowed.  The next step is to fill in the
tables with the different indices of the tensor.  In particular, we
write
\begin{equation}\label{irr2}
\Yvcentermath1\young(ij)=\Big(I+(ij)\Big)A_{ij}=A_{ij}+A_{ji}\,,
\end{equation}
and
\begin{equation}\label{irr2x}
 \Yvcentermath1\young(i,j)=\Big(I-(ij)\Big)A_{ij}=A_{ij}-A_{ji}\,.
\end{equation}
Here $I$ denotes the identity operator, $(ij)$ is the permutation
operator, $\big(I+(ij)\big)$ and $\big(I-(ij)\big)$ are respectively
Young's symmetrizer and antisymmetrizer operators.  Thus, we have
\begin{equation}\label{irr1x}
  \Yvcentermath1\young(i)\otimes\young(j)=\a\,\young(ij)\oplus \b\,
\young(i,j)\,.
\end{equation}
The decomposition (\ref{irr}) is then obtained by inserting suitable
leading coefficients $\a$ and $\b$. In general, these coefficients are
calculated by using a combinatorial formula. It is easy to see that,
for the completely symmetric and completely antisymmetric diagrams of
$n$ cells, the coefficients are equal to $1/n!$. Thus, the coefficients
in (\ref{irr1x}) are equal to $1/2$ and the decomposition (\ref{irr})
is recovered.

\subsection{Young's decomposition of fourth rank tensors}

The first step is to construct Young's tableaux corresponding to a
fourth rank tensor.
\begin{rul}
  The four cells representing the indices of the tensor must be glued
  into tables of all possible shapes.  The only restriction is that
  the number of cells in any row must be less or equal to the number
  of cells in the previous row. The number of irreducible subtensors
  of a given tensor is equal to the number of the tables of all
  possible shapes.
\end{rul}
\noindent Due to this rule, a generic fourth order tensor can be
irreducibly decomposed into the sum of five independent parts. These
parts are described by the following Young's diagrams, which are the
graphical representation of the permutation group $S_4$,
{see Boerner (1970)\cite{Boerner} or} Hamermesh
  (1989)\cite{hamermesh},
\begin{equation}\label{irr3}
  \Yvcentermath1\yng(1)\otimes\yng(1)\otimes\yng(1)\otimes\yng(1)
  =\yng(4)\oplus \yng(3,1) \oplus \yng(2,2)\oplus \yng(2,1,1)
  \oplus \yng(1,1,1,1)\,.
\end{equation}
The left-hand side describes a generic fourth order tensor.  On the
right-hand side, the first diagram represents the completely symmetric
tensor. The middle diagrams are for the tensors which are partially
symmetric and partially antisymmetric. The last diagram represents a
completely antisymmetric tensor.

The next step of Young's procedure is to fill in the tables with the
indices. In order to avoid the repetitions, the following rule is
used:
\begin{rul}
  In each row and each column of Young's table, the positions of the
  indices, i.e., the numbers $i_1,i_2\cdots,i_p$, are inserted in the
  increasing order. The symmetrization operators correspond to the
  rows of the table.  Since different rows contain no common indices,
  the corresponding symmetrization operators commute.  Different
  antisymmetrization operators correspond to the columns of the
  table. They also commute with each other and can be taken in an
  arbitrary order. The symmetrization and antisymmetrization operators
  act on the same indices, hence they do not commute and thus must be
  taken in some fixed order.
\end{rul}
Thus, due to graphical representation (\ref{irr3}) of the permutation
group $S_4$, we have the following parts of a generic three
dimensional tensor $T_{\i\j\k\l}$
\begin{eqnarray}\label{tab1}
\hspace{-20pt}  \Yvcentermath1\young(\i)\otimes\young(\j)\otimes\young(\k)\otimes
  \young(\l)&=&\Yvcentermath1 \a\,\young({\i}{\j}{\k}{\l})
  \Yvcentermath1+\b\Bigg(\young({\i}{\j}{\k},{\l})+ 
  \young({\i}{\j}{\l},{\k})+ \young({\i}{\k}{\l},{\j})\Bigg)
  \Yvcentermath1+\nonumber\\
  &&\Yvcentermath1\g\Bigg( \,\young({\i}{\j},{\k}{\l})+ 
  \young({\i}{\k},{\j}{\l})\,\Bigg)+\Yvcentermath1 \d\Bigg(\, \young({\i}{\j},{\k},{\l})+ 
  \young({\i}{\k},{\j},{\l})+\young({\i}{\l},{\j},{\k})\,\Bigg)
  +\varepsilon\,\young({\i},{\j},{\k},{\l})\,.
\end{eqnarray}
\begin{rul}
  The coefficients $\a,\b,\g,\d,\varepsilon$ are determined by the
  combinatorial formula, see Hamermesh (1989)\cite{hamermesh}.  They
  can be calculated using the {\it Mathematica} package {\it
    'Combinatorics'}. The first and the last coefficients are
  especially simple. For an $n$-order tensor,
\begin{equation}\label{coeff}
\a=\varepsilon=\frac 1{n!}\,.
\end{equation}
\end{rul}
For a generic 4th rank tensor, all these diagrams are relevant.  In an
$n$-dimensional space, it has $n^4$ independent components, which are
distributed between the diagrams (\ref{tab1}). The explicit form of
the corresponding terms can be found in Wade (1941)\cite{wade}.

\end{document}